\documentclass[12pt,oneside,english]{amsart}
\usepackage[T1]{fontenc}
\usepackage[latin9]{inputenc}
\usepackage{geometry}
\usepackage{units}
\usepackage{amstext}
\usepackage{amsthm}
\usepackage{amssymb}
\usepackage{setspace}
\makeatletter
\theoremstyle{plain}
\newtheorem{thm}{\protect\theoremname}[section]
  \theoremstyle{plain}
  \newtheorem{prop}[thm]{\protect\propositionname}
  \theoremstyle{definition}
  \newtheorem{example}[thm]{\protect\examplename}
  \theoremstyle{remark}
  \newtheorem{rem}[thm]{\protect\remarkname}
  \theoremstyle{plain}
  \newtheorem{lem}[thm]{\protect\lemmaname}
  \theoremstyle{plain}
  \newtheorem{cor}[thm]{\protect\corollaryname}
  \theoremstyle{definition}
  \newtheorem{problem}[thm]{\protect\problemname}
  \theoremstyle{plain}
  \newtheorem{conjecture}[thm]{\protect\conjecturename}

\usepackage{hyperref}
\usepackage{breqn}
\usepackage{bbm}

\DeclareMathOperator{\sgn}{sgn} 
\DeclareMathOperator{\Stab}{Stab} 
\DeclareMathOperator{\NS}{NS} 
\DeclareMathOperator{\Inf}{Inf}

\DeclareMathOperator{\Lev}{Lev}

\DeclareMathOperator{\Deg}{Deg}
\DeclareMathOperator{\Dist}{Dist}
\DeclareMathOperator{\Chow}{Chow}
\DeclareMathOperator{\Gap}{Gap}

\DeclareMathOperator{\maj}{\mathrm{Maj}}

\DeclareMathOperator{\edic}{\mathrm{E-Dict}}
\DeclareMathOperator{\OR}{OR}

\allowdisplaybreaks

\global\long\def\s[#1]{\mathrm{\scriptsize #1}}

\global\long\def\P{\Pr}
\global\long\def\E{\mathbb{E}}

\global\long\def\I{\mathbbm{1}}
\global\long\def\d{\mathrm{d}}

\global\long\def\binent{h}
\global\long\def\bindiv{d}

\global\long\def\Hamd{d_H}

\global\long\def\dfn{:=}
\global\long\def\trre[#1,#2]{\overset{{\scriptstyle (#2)}}{#1}} % transition explained with reason

\author{Nir Weinberger and Ofer Shayevitz}
\thanks{The authors are with the Department of EE--Systems, Tel Aviv University, Tel Aviv, Israel. Emails: \{nir.wein@gmail.com, ofersha@eng.tau.ac.il\}. This work was supported by an ERC grant no. 639573. The material in this paper was presented in part at the Information Theory and Applications (ITA) Workshop, San Diego, California, U.S.A., February 2018, and at the IEEE International Symposium on Information Theory (ISIT), Vail, Colorado, U.S.A., June 2018}

\makeatother

\usepackage{babel}
  \providecommand{\conjecturename}{Conjecture}
  \providecommand{\corollaryname}{Corollary}
  \providecommand{\examplename}{Example}
  \providecommand{\lemmaname}{Lemma}
  \providecommand{\problemname}{Problem}
  \providecommand{\propositionname}{Proposition}
  \providecommand{\remarkname}{Remark}
\providecommand{\theoremname}{Theorem}

\begin{document}

\title[Self-Predicting Boolean Functions]{Self-Predicting Boolean Functions}


\begin{abstract}
A Boolean function $g$ is said to be an \emph{optimal predictor}
for another Boolean function $f$, if it minimizes the probability
that $f(X^{n})\neq g(Y^{n})$ among all functions, where $X^{n}$
is uniform over the Hamming cube and $Y^{n}$ is obtained from $X^{n}$
by independently flipping each coordinate with probability $\delta$.
This paper is about \emph{self-predicting functions}, which are those
that coincide with their optimal predictor. 
\end{abstract}

\maketitle

\section{Introduction}

One of the most important properties of a Boolean function $f:\{-1,1\}^{n}\to\{-1,1\}$
is its robustness to noise in its inputs. This robustness is traditionally
measured by the \emph{noise sensitivity} of the function 
\begin{equation}
\NS_{\delta}[f]\dfn\P\left(f(X^{n})\neq f(Y^{n})\right),\label{eq: noise sensitivity}
\end{equation}
where $X^{n}\in\{-1,1\}^{n}$ is a uniform Bernoulli vector, and $Y^{n}\in\{-1,1\}^{n}$
is obtained from $X^{n}$ be flipping each coordinate independently
with probability $0<\delta<1/2$. The noise sensitivity of Boolean
functions has been extensively investigated \cite{Bool_book}, most
often in terms of the equivalent notion of \emph{stability} 
\[
\Stab_{\rho}[f]\dfn\E\left[f(X^{n})f(Y^{n})\right],
\]
where $0<\rho<1$ is the correlation parameter, i.e., $\rho\dfn\E(X_{i}Y_{i})=1-2\delta$.
The noise sensitivity of $f$ can also be interpreted as the error
probability of a predictor trying to guess the value of $f(X^{n})$
by simply applying $f$ to the noisy input $Y^{n}$. While this predictor
is intuitively appealing and easy to analyze, it is generally suboptimal.
As a simple example, think of the case where $f$ is biased and the
noise level $\delta$ is sufficiently high; it is easy to see that
a constant predictor would result in a lower error probability than
$f(Y^{n})$ would. 

The optimal predictor, i.e., the one that minimizes the error probability
in predicting $f(X^{n})$ from $Y^{n}$, is clearly given by the sign
of $\E(f(X^{n})\mid Y^{n}=y^{n})$. In general, this function might
be rather different from $f$ itself. However, while using the optimal
predictor is generally superior to using the function itself (albeit
as we shall see, by a factor of two at the most), computing the former
is often very difficult as its value in an point depends on the values
of the function over the entire Hamming cube. It is therefore interesting
to study functions that coincide with their optimal predictor; we
call these functions \emph{self-predicting} (SP). 

Clearly, SP functions exhibit a desirable property - the optimal prediction
of the function is obtained by simply applying it to the noisy inputs.
For example, suppose the function describes a voting rule and the
noise represents possible contamination of the votes (e.g., due to
fraud). If the function is SP, then any mechanism used for computing
the function with clean votes can be used without any modification
in case it turns out that the votes are actually noisy. In this case,
the output of the function is the optimal predictor for its true value.
It should be noted, however, that being SP does not imply anything
about the ordinary stability of the function. For example, all parity
functions (characters) are SP functions, including the least stable
one, to wit, the parity of all inputs (namely, the largest character).
Nonetheless, if, e.g., there are a few alternatives for choosing a
function to be used, and all of these functions have the same stability,
it is sensible to choose one of the SP functions among them (if such
exists).

Nonetheless, a function can be SP at certain noise levels but not
at others. We thus say that a function is \emph{uniformly SP (USP)}
if it is SP at any noise level. For example, in the voting scenario
mentioned above, it may not be realistic to assume that the noise
level is known, yet if the function is USP it can always be used to
obtain the optimal prediction of the true voting result.

In this paper, we introduce and explore self-predictability of Boolean
functions. We derive various properties of SP functions, and specifically
the following:
\begin{itemize}
\item If a function is monotone (resp. odd, resp. symmetric), then so is
the optimal predictor. We use this fact to show that Majority functions
are USP, and that for a monotone function, self-predictability at
dominating boundary points is necessary and sufficient for a function
to be SP. 
\item SP at high correlation: A function with Fourier degree $k$ is SP
for any $\rho>1-1/k^{2}$, and if $f$ is SP for $\rho>1-\varepsilon$
and $n=\Omega(1/\varepsilon)$, then each point $x^{n}$ has a distance-$2$
neighbor with the same function value.
\item SP at low correlation: Any function for which there exists $\rho^{*}$
such that it is SP for all $\rho\in[0,\rho^{*}]$ (abbreviated LCSP)
is spectral threshold, i.e., equal to the sign of its lowest Fourier
level. This simple fact implies many properties: LCSP functions are
either balanced or constant, they have energy at least $1/2$ on their
first level (if any), and a monotone LCSP function is $\sqrt{\frac{2}{\pi n}}$-close
to a linear threshold function.
\item Sharp threshold: All functions are trivially SP for $\rho>1-\frac{2\ln2}{n}+O(n^{-2})$.
However, the fraction of SP functions is doubly-exponential small
with $n$ whenever the correlation parameter is either $\rho=1-\frac{2\alpha}{n}$
for $\alpha>1$ or $\rho=1-2\delta$ for $\delta\in[0,\delta_{\max}]$,
$\delta_{\max}\approx0.097$.
\end{itemize}
The paper is organized as follows. Section \ref{sec:Definitions}
contains basic notation and Fourier-theory facts. The self-predictability
problem and some basic properties are introduced in Section \ref{sec: USP functions},
including the proof that Majority is USP. Section \ref{sec:high-corr}
discusses high-correlation sufficient conditions for SP, and Section
\ref{sec:low-corr} discusses low-correlation SP functions. Section
\ref{sec:stab} provides stability-based necessary conditions for
SP. In Section \ref{sec:Sharp-Threshold}, a \emph{sharp threshold}
phenomenon is proved for the SP property. The paper is concluded in
Section \ref{sec:Open Problems} with a list of open problems.

\section{Preliminaries\label{sec:Definitions}}

\subsection{Notation and Definitions}

We use upper case letters for random variables and random vectors,
and their lower case counterparts for specific realizations. For vectors
we write $x_{i}^{j}=(x_{i},\ldots,x_{j})$ and omit the subscript
whenever $i=1$, and denote a concatenation of vectors by $(x_{i}^{j},x_{k}^{m})=(x_{i},\ldots,x_{j},x_{k},\ldots,x_{m})$.
We denote the cardinality of a set $S$ by $|S|$, the complement
of the set $\mathcal{A}$ by $\mathcal{A}^{c}$, and write $[n]$
for the set $\{1,2,\ldots,n\}$. We define the indicator function
by $\I(\cdot)$, the sign function by $\sgn(z)$ where by convention
$\sgn(0)=0$, unless otherwise stated. Throughout, the logarithm $\log(t)$
is base $2$, while $\ln(t)$ is the natural logarithm. The Hamming
distance between $x^{n}$ and $y^{n}$ is $\Hamd(x^{n},y^{n})$. 

In this paper, $X^{n}$ is a uniformly distributed binary vector,
and $Y^{n}$ is the binary vector obtained by flipping each coordinate
of $X^{n}$ with some given probability $\delta\in[0,1/2]$. We write
$p(x^{n},y^{n})$ to denote the associated joint probability mass
function, and $p(x^{n}\mid y^{n})$, e.g., to denote the conditional
probability mass function. As a binary alphabet, for the most part
we will find it convenient to work with $\{-1,1\}$, in which case
it is more natural to consider the \emph{correlation} parameter $\rho\dfn\E(X_{i}Y_{i})=1-2\delta\in[0,1]$
instead of the crossover probability parameter $\delta$. We will
use the latter notations throughout the paper, with the exception
of a few proofs where we find it more convenient to work with either
$\delta$ or the binary alphabet $\{0,1\}$.

\subsection{Boolean Functions and Fourier Analysis}

In this paper we consider Boolean functions $f:\{-1,1\}^{n}\to\{-1,1\}$.
The \emph{distance} between two Boolean functions $f$ and $g$ is
defined as the fraction of inputs on which they disagree, i.e., $\P(f(X^{n})\neq g(X^{n}))$.
We say that $f$ and $g$ are \emph{$\varepsilon$-close} if their
distance is at most $\varepsilon$.

An inner product between two Boolean functions $f,g$ is defined as
\begin{equation}
\left\langle f,g\right\rangle \dfn\E\left(f(X^{n})g(X^{n})\right).\label{eq: inner product}
\end{equation}
A \emph{character} associated with a set of coordinates $S\subseteq[n]$
is the Boolean function $x^{S}\dfn\prod_{i\in S}x_{i}$, where by
convention $x^{\emptyset}=1$. It can be shown \cite[Chapter 1]{Bool_book}
that the set of all characters form an orthonormal basis with respect
to (w.r.t.) the inner product (\ref{eq: inner product}). Furthermore,
\[
f(x^{n})=\sum_{S\subseteq[n]}\hat{f}_{S}\cdot x^{S},
\]
where $\{\hat{f}_{S}\}_{S\subseteq[n]}$ are the \emph{Fourier coefficients
}of $f$, given by $\hat{f}_{S}=\langle x^{S},f\rangle=\E(X^{S}\cdot f(X^{n}))$.
When $S$ is a singleton $\{i\}\subset[n]$, we use the shorthand
$\hat{f}_{i}=\hat{f}_{\{i\}}$. The \emph{Fourier weight} of $f$
at degree $k$ is 
\[
W^{k}[f]\dfn\sum_{S\subseteq[n]\colon|S|=k}\hat{f}_{S}^{2}.
\]

Instead of the \emph{noise sensitivity} defined in (\ref{eq: noise sensitivity})
it is more common to consider the \emph{stability}, defined as 
\[
\Stab_{\rho}[f]\dfn\E\left(f(X^{n})f(Y^{n})\right),
\]
where the noise sensitivity and stability are trivially related via
\[
\Stab_{\rho}[f]=1-2\NS_{\frac{1-\rho}{2}}[f].
\]
Thus, the stability of a function is directly related to the error
probability of the possibly suboptimal predictor $f(y^{n})$ to the
function's true value $f(x^{n})$. 

The \emph{noise operator} for \emph{$\rho$}-correlated $X^{n}$ and
$Y^{n}$ is defined as
\[
T_{\rho}f(y^{n})\dfn\E\left(f(X^{n})\mid Y^{n}=y^{n}\right),
\]
and, evidently, as $\{(X_{i},Y_{i})\}$ is an i.i.d. sequence, 
\begin{equation}
T_{\rho}f(y^{n})=\E\left(\sum_{S\subseteq[n]}\hat{f}_{S}\cdot X^{S}\mid Y^{n}=y^{n}\right)=\sum_{S\subseteq[n]}\rho^{|S|}\cdot\hat{f}_{S}\cdot y^{S}.\label{eq: Fourier expansion of noise operator}
\end{equation}
The stability can then be expressed using the Fourier coefficients
and the noise operator as 
\begin{align}
\Stab_{\rho}[f] & =\E\left(\E\left(f(X^{n})f(Y^{n})\right)\mid Y^{n}\right)\nonumber \\
 & =\E\left(f(Y^{n})T_{\rho}f(Y^{n})\right)\nonumber \\
 & =\langle f,T_{\rho}f\rangle\nonumber \\
 & =\sum_{S\subseteq[n]}\rho^{|S|}\cdot\hat{f}_{S}^{2}\label{eq: Stability as Fourier}\\
 & =\left\Vert T_{\sqrt{\rho}}f\right\Vert _{2}^{2},\nonumber 
\end{align}
where (\ref{eq: Stability as Fourier}) follows from \emph{Plancherel's
identity }$\langle f,g\rangle=\E(f(X^{n})g(X^{n}))=\sum_{S\subseteq[n]}\hat{f}_{S}\hat{g}_{S}$.

A Boolean function $f$ is called a \emph{linear threshold function
(LTF)} if there exists coefficients $a_{0}^{n}\in\mathbb{R}^{n+1}$
such that 
\[
f(x^{n})=\sgn\left(a_{0}+\sum_{i=1}^{n}a_{i}x_{i}\right).
\]
Note that if $a_{0}=0$ then $f$ is \emph{balanced}, i.e., $\Pr(f(X^{n})=1)=1/2$.
More generally, a function $f$ is a \emph{polynomial threshold function}
(PTF) \cite{bruck1990harmonic} of degree $k$ if there exists $\{\hat{p}_{S}\}$
such that $\max_{S:\hat{p}_{S}\neq0}|S|=k$ and 
\begin{equation}
f(x^{n})=\sgn\left(\sum_{S\subseteq[n]}\hat{p}_{S}\cdot x^{S}\right).\label{eq: PTF}
\end{equation}
A PTF has \emph{sparsity} $s$ if $\{\hat{p}_{S}\}$ is supported
over exactly $s$ terms. For LTF and PTFs, we will always assume that
coefficients are chosen such that the polynomial inside the sign operator
is never identically zero.

\section{Optimal Prediction and Self Predicting (SP) Functions \label{sec: USP functions}}

\subsection{The Optimal Predictor}

Let $f:\{-1,1\}^{n}\to\{-1,1\}$ be some Boolean function. It is easy
to see that the optimal predictor (minimizing the error probability)
of $f(X^{n})$ given that $Y^{n}=y^{n}$ has been observed, is simply
\[
\sgn\E\left(f(X^{n})\mid Y^{n}=y^{n}\right)=\sgn T_{\rho}f(y^{n}).
\]
Note that according to our definition $\sgn(0)=0$, but ties can of
course be broken arbitrarily in any other way. 

The optimal predictor preserves several properties of the function.
We define the natural partial order $\preceq$ over $\mathbb{R}^{k}$,
where $y^{k}\preceq z^{k}$ if and only if $y_{i}\leq z_{i}$ for
all coordinates $i$. We write $\prec$ to denote the case of strict
inequality in at least one of the coordinates. 

Recall that \cite[Definition 2.8.]{Bool_book} a function $f:\{-1,1\}^{n}\to\mathbb{R}$
is called:
\begin{itemize}
\item \emph{Monotone on} $S\subseteq[n]$, if $f(y^{n})\leq f(z^{n})$ whenever
both $y^{S}\preceq z^{S}$ and $y^{[n]\setminus S}=z^{[n]\setminus S}$,
and \emph{monotone }if it is monotone on $[n]$.
\item \emph{Odd (resp. even) if} $f(x^{n})=-f(-x^{n})$ for all $x^{n}\in\{-1,1\}^{n}$
(resp. $f(x^{n})=f(-x^{n})$).
\item \emph{Symmetric }if $f(\pi(x^{n}))=f(x^{n})$ for all $x^{n}\in\{-1,1\}^{n}$
and permutation $\pi\in\mathbb{S}_{n}$ (where $\mathbb{S}_{n}$ is
the symmetric group over the set $[n]$) and $\pi(x^{n})=(x_{\pi(1)},\ldots,x_{\pi(n)})$.
\end{itemize}
\begin{prop}
\label{prop: Optimal predictor preserves OSM}For $\rho\in(0,1]$,
$\sgn T_{\rho}(\cdot)$ preserves monotonicity on any $S\subseteq[n]$,
parity (oddness or evenness) and symmetry.
\end{prop}
\begin{proof}
~
\begin{itemize}
\item \emph{Monotonicity}: This property stems from the fact that the operator
$T_{\rho}$ itself preserves monotonicity (for $\rho\in(0,1]$) \cite[Proof of Proposition 4.4]{keller2009probability},
\cite[Claim 2.4. (b)]{kalai2016correlation}. A short proof is given
for the sake of completeness. Assume that $f(y^{n})=1$ and let $z^{n}$
satisfy $y^{S}\preceq z^{S}$ and $y^{[n]\setminus S}=z^{[n]\setminus S}$.
We prove the statement for a singleton $S$, say $S=\{n\}$. The general
case then follows by applying the same argument repeatedly. If $y_{n}=1$
the claim is trivial. Assume $y_{n}=-1$ and let $z^{n}$ agree with
$y^{n}$ except on the $n$th coordinate. Due to monotonicity of $f$,
we have that $f(z^{n})=1$. Then 
\begin{align*}
T_{\rho}f(z^{n}) & =\sum_{x^{n}}p(x^{n}\mid z^{n})f(x^{n})\\
 & =\sum_{x^{n-1}}\sum_{x_{n}}p(x^{n-1}\mid y^{n-1})p(x_{n}\mid1)f(x^{n})\\
 & =\sum_{x^{n-1}}p(x^{n-1}\mid y^{n-1})\left[\delta f(x^{n-1},-1)+(1-\delta)f(x^{n-1},1)\right]\\
 & \geq\sum_{x^{n-1}}p(x^{n-1}\mid y^{n-1})\left[(1-\delta)f(x^{n-1},-1)+\delta f(x^{n-1},1)\right]\\
 & =T_{\rho}f(y^{n})
\end{align*}
where the inequality holds since $f$ is monotone on the $n$th coordinate
(and $\delta\in[0,1/2)$). Hence, $\sgn T_{\rho}f(z^{n})\geq\sgn T_{\rho}f(y^{n})$. 
\item \emph{Parity}: $f$ is odd if and only if $\hat{f}_{S}=0$ for all
$S\subseteq[n]$ such that $|S|$ is even \cite[Exercise 1.8]{Bool_book}.
It follows from the Fourier expansion of $T_{\rho}f$ (\ref{eq: Fourier expansion of noise operator})
that if $f$ is odd then so is $T_{\rho}f$, i.e. $T_{\rho}f(x^{n})+T_{\rho}f(-x^{n})=0$
for all $x^{n}\in\{-1,1\}^{n}$. Thus, $\sgn T_{\rho}f$ is also odd
(utilizing the convention $\sgn(0)=0$). The proof for even functions
is similar.
\item \emph{Symmetry}: $f$ is symmetric if and only if $\hat{f}_{S}$ depends
on $S$ only via $|S|$. Hence (\ref{eq: Fourier expansion of noise operator})
implies that if $f$ is symmetric then so is $T_{\rho}f$. A composition
of scalar function and a symmetric function results in a symmetric
function, and thus $\sgn T_{\rho}f$ is symmetric.
\end{itemize}
\end{proof}
We say that a Boolean function $f$ is \emph{$\rho$-self-predicting
($\rho$-SP) at $y^{n}$}, if the optimal predictor given $y^{n}$
at correlation level $\rho$ coincides with the function itself whenever
it is not tied, i.e., if 
\[
f(y^{n})=\sgn T_{\rho}f(y^{n}),
\]
whenever $T_{\rho}f(y^{n})\neq0$. The function $f$ is called $\rho$-SP
if it is $\rho$-SP for any $y^{n}\in\{-1,1\}^{n}$. We say that $f$
is \emph{uniformly self-predicting (USP)} if it is $\rho$-SP for
any $\rho\in[0,1]$. We also say that $f$ is \emph{low-correlation
self-predicting (LCSP)}, if there exists some $\rho^{*}>0$ such that
$f$ is $\rho$-SP for all $\rho\in[0,\rho^{*})$.

We note in passing that seemingly plausible properties may not hold
in general:
\begin{example}
\label{exa:balanced function may not have balanced predictor}The
optimal predictor of a balanced function may not be balanced. For
example, the function
\begin{align*}
 & \frac{1}{4}(2x_{1}+x{}_{3}-2x{}_{1}x{}_{2}+x{}_{1}x{}_{3}+x{}_{2}x{}_{3}-x{}_{3}x{}_{4}\\
 & \hphantom{=}+x{}_{1}x{}_{2}x{}_{3}+x{}_{1}x{}_{3}x{}_{4}-x{}_{2}x{}_{3}x{}_{4}+x{}_{1}x{}_{2}x{}_{3}x{}_{4})
\end{align*}
is a balanced function, yet $\sgn T_{\rho}f$ is non-balanced when
$\rho=1/2$.
\end{example}
\begin{example}
\label{exa: LTF example}In the following sections we explore functions
that are SP for high or low correlation. However, self-predictability
is not necessarily a monotone property in $\rho$. To wit, if a function
is $\rho_{0}$-SP then it might not be $\rho$-SP for some $\rho\geq\rho_{0}$.
Indeed, there are functions that admit such an ``irregular'' behavior.
We have numerically analyzed LTFs with randomly drawn coefficients,
and found, for example, that the balanced LTF with $n=11$ and coefficients
\[
a_{1}^{11}=(13,43,67,67,67,117,153,165,165,179,179)
\]
is $\rho$-SP only for $\rho\in[0,0.312]\cup(0.544,1]$. 
\end{example}

\subsection{Elementary USP Functions}

The following fact follows easily from the definition. 
\begin{prop}
\label{prop: Characters are USP}All the characters are USP. 
\end{prop}
\begin{proof}
Let $f(x^{n})=x^{S}$ for some $S\subseteq[n]$. Then for any $y^{n}$,
\begin{align*}
\sgn T_{\rho}f(y^{n}) & =\sgn\left(\rho^{|S|}\cdot y^{S}\right)\\
 & =\sgn\left(y^{S}\right)\\
 & =f(y^{n}).
\end{align*}
\end{proof}
We next show that Majority (for odd $n$), given by, 
\[
\maj(x^{n})\dfn\sgn\sum_{i\in[n]}x_{i}
\]
is USP. While this property is plausible, it does not stem from only
analyzing the ``local'' behavior of the function. Specifically,
at a boundary point $y^{n}$, i.e., one for which $\sum_{i\in[n]}y_{i}=\pm1$,
there are more neighbors in the immediate neighborhood of $y^{n}$
(say, Hamming distance one or two) who disagree with $y^{n}$ on the
value of the function, than those who agree with it. Thus, any proof
that such a point is SP for all $\rho\in(0,1]$ cannot rely only on
the local values of the function in the vicinity of that point. Rather,
it should take into account the function's value in larger neighborhoods,
or even over the entire Hamming cube.
\begin{thm}
\label{thm: Majority USP}Majority is USP. 
\end{thm}
\begin{proof}
Since $\maj$ is monotone, odd and symmetric, then so is $\sgn T_{\rho}\maj$
(Proposition \ref{prop: Optimal predictor preserves OSM}). Hence,
for all $x^{n}\in\{-1,1\}^{n}$
\begin{equation}
\sgn T_{\rho}\maj(x^{n})+\sgn T_{\rho}\maj(-x^{n})=0.\label{eq: sgn T_rho majority is odd explicit}
\end{equation}
Consider without loss of generality $x^{n}$ such that $\maj(x^{n})=1$,
i.e., if $w$ is the number of $1$'s in $x^{n}$, then $w>n-w$.
Then, $\overline{x}=(1^{w},-1^{n-w})$ and $\tilde{x}=(1^{n-w},-1^{w})$
satisfy $\tilde{x}^{n}\preceq\overline{x}^{n}$, and from symmetry,
\[
\sgn T_{\rho}\maj(x^{n})=\sgn T_{\rho}\maj(\overline{x}^{n})\geq\sgn T_{\rho}\maj(\tilde{x}^{n})=\sgn T_{\rho}\maj(-x^{n}).
\]
Hence, (\ref{eq: sgn T_rho majority is odd explicit}) implies that
$\sgn T_{\rho}\maj(x^{n})\geq0$, as was required to be proved.
\end{proof}
\begin{rem}
\label{rem: proof of majority is USP via May's theorem}An indirect
way of proving Theorem \ref{thm: Majority USP} is via May's theorem
\cite[Ex. 2.3]{Bool_book}: Since $\sgn T_{\rho}\maj$ is monotone,
odd and symmetric, it must be the majority function itself.
\end{rem}
By numerically experimenting with simple LTFs one can find that Majority
(and characters) are not the only USP functions, and not even the
only USP LTFs. Specifically:
\begin{example}
\label{exa:USP not majority}The balanced LTFs with $n=5$ and coefficients
$a_{1}^{5}=(1,1,3,3,5)$, with $n=7$ and coefficients $a_{1}^{7}=(1,1,3,3,3,5,7)$,
with $n=9$ and coefficients $a_{1}^{9}=(1,1,3,3,3,5,5,5,7)$, with
$n=11$ and coefficients $a_{1}^{11}=(1,1,3,3,3,3,5,5,5,7,7)$ can
all be verified by direct computation to be USP. 

In the next section we generate classes of USP functions by utilizing
operations which preserve the SP property.
\end{example}

\subsection{SP/USP Preserving Operators}

Let us next discuss several operations that preserve self-predictability.
First, we note that self-predictability is invariant to negation of
inputs. We write $\circ$ for the Hadamard product. 
\begin{prop}
\label{prop:Input flip preserves SP}Let $a^{n}\in\{-1,1\}^{n}$.
Then, $f(x^{n})$ is $\rho$-SP if and only if $f(a^{n}\circ x^{n})$
is $\rho$-SP. 
\end{prop}
The straightforward proof is omitted. Next, we consider the case of
separable functions. 
\begin{prop}
\label{prop:seperable functions}Let $f(x^{n})=g(x_{1}^{k})\cdot h(x_{k+1}^{n})$.
Then $f$ is $\rho$-SP if and only if both $g$ and $h$ are $\rho$-SP. 
\end{prop}
\begin{proof}
If $g$ and $h$ are both $\rho$-SP then for any $y^{n}$, 
\begin{align*}
\sgn T_{\rho}f(y^{n}) & =\sgn T_{\rho}\left(g(y^{k})\cdot h(y_{k+1}^{n})\right)\\
 & =\sgn\left(T_{\rho}g(y^{k})\cdot T_{\rho}h(y_{k+1}^{n})\right)\\
 & =g(y_{1}^{k})\cdot h(y_{k+1}^{n})\\
 & =f(y^{n}).
\end{align*}

Conversely, suppose that $f$ is $\rho$-SP for some $\rho\in(0,1]$.
Then there must exist at least one point $y_{k+1}^{n}$ at which $h$
is $\rho$-SP, since if this was not the case, then $T_{\rho}h(y_{k+1}^{n})\cdot f(h_{k+1}^{n})\leq0$
holds for all $h_{k+1}^{n}$. This, however, is impossible since
\[
\E\left[T_{\rho}h(Y_{k+1}^{n})\cdot h(Y_{k+1}^{n})\right]=\sum\rho^{|S|}\hat{h}_{S}^{2}>0.
\]
Hence, without loss of generality, we may assume that $h(y_{k+1}^{n})=1$.
Then for any $y^{k}$ 
\begin{align*}
\sgn T_{\rho}g(y^{k}) & =\sgn T_{\rho}g(y^{k})\cdot\sgn T_{\rho}h(y_{k+1}^{n})\\
 & =\sgn\left(T_{\rho}g(y^{k})\cdot T_{\rho}h(y_{k+1}^{n})\right)\\
 & =\sgn T_{\rho}f(y^{n})\\
 & =f(y^{n})\\
 & =g(y^{k})\cdot h(y_{k+1}^{n})\\
 & =g(y^{k}).
\end{align*}
Hence $g$, and symmetrically, also $h$, are $\rho$-SP. 
\end{proof}
Note that Proposition \ref{prop: Characters are USP} also follows
as a simple corollary to Proposition \ref{prop:seperable functions}.
Next, we consider functions of equal-size disjoint characters.\textbf{
} 
\begin{prop}
\label{prop:majority of independent characters}Let $\{S_{\ell}\subseteq[n]\}_{\ell\in[m]}$
be disjoint subsets of equal size $|S_{\ell}|=w$. Let $f:\{-1,1\}^{m}\to\{-1,1\}$
be $\rho^{w}$-SP. Then $f(x^{S_{1}},x^{S_{2}},\ldots,x^{S_{m}})$
is $\rho$-SP. 
\end{prop}
\begin{proof}
By equating coefficients of the Fourier representation (which are
unique), it is readily obtained that the Fourier coefficients of $h(x^{n})=f(x^{S_{1}},x^{S_{2}},\ldots,x^{S_{m}})$
are given by 
\[
\hat{h}_{S}=\begin{cases}
\hat{f}_{T}, & S=\cup_{t\in T}S_{t}\\
0, & \text{otherwise}
\end{cases}.
\]
Hence, 
\begin{align*}
\sgn T_{\rho}h(y^{n}) & =\sgn\sum_{S\subseteq[n]}\rho^{|S|}\hat{h}_{S}y^{S}\\
 & =\sgn\sum_{T\subseteq[m]}\rho^{w|T|}\cdot\hat{h}_{\cup_{t\in T}S_{t}}\cdot y^{\cup_{t\in T}S_{t}}\\
 & =\sgn\sum_{T\subseteq[m]}\rho^{w|T|}\hat{f}_{T}\prod_{t\in T}y^{S_{t}}\\
 & =\sgn T_{\rho^{w}}f(y^{S_{1}},y^{S_{2}},\ldots,y^{S_{m}})\\
 & =f(y^{S_{1}},y^{S_{2}},\ldots,y^{S_{m}})\\
 & =h(y^{n}).
\end{align*}
\end{proof}
\begin{example}
\label{exa:Majority characters and operations}Using the fact that
characters and Majority are USP functions, together with Propositions
\ref{prop:Input flip preserves SP}, \ref{prop:seperable functions}
and \ref{prop:majority of independent characters}, we can construct
many distinct USP functions. For example, the function 
\[
\sgn\left((x_{1}x_{2}+x_{3}x_{4}+x_{5}x_{6})\cdot(x_{7}x_{8}x_{9}-x_{10}x_{11}x_{12}-x_{13}x_{14}x_{15})\cdot x_{16}\right)
\]
is USP.
\end{example}
Nonetheless, there are USP functions that cannot be constructed from
characters and Majority this way. For example, none of these functions
can be an LTF, as the USP functions in Example \ref{exa:USP not majority}. 

\subsection{Closeness to SP and Strong Stability}

How far can a function be from self predicting? We say that a function
is \emph{$\varepsilon$-close to $\rho$-SP}, to mean that $f$ and
its optimal predictor $\sgn T_{\rho}f$ are $\varepsilon$-close. 
\begin{lem}
\label{lem: Any function is SP for some output}Any function $f$
is $\sum_{S\subseteq[n]}(1-\rho^{|S|})\hat{f}_{S}^{2}$-close to $\rho$-SP.
 
\end{lem}
\begin{proof}
Let $A\subseteq\{-1,1\}^{n}$ be the set of all $y^{n}$ at which
$f$ is $\rho$-SP. Hence for any $y^{n}\not\in A$ it must be that
$f(y^{n})\cdot T_{\rho}f(y^{n})<0$. Noting that $|T_{\rho}f(y^{n})|\leq1$,
we have that 
\[
\E\left(f(Y^{n})\cdot T_{\rho}f(Y^{n})\right)\leq\P\left(Y^{n}\in A\right).
\]
On the other hand, it also holds that 
\[
\E\left(f(Y^{n})\cdot T_{\rho}f(Y^{n})\right)=\sum_{S\subseteq[n]}\rho^{|S|}\hat{f}_{S}^{2}.
\]
The proof now follows by recalling that $\sum_{S\subseteq[n]}\hat{f}_{S}^{2}=1$. 
\end{proof}
For any $n$, functions that depend on all $n$ variables can be found
(even balanced ones), whose distance from their optimal predictor
is larger than some universal constant. The problem with this measure
of closeness to SP is that in many cases the optimal predictor might
be different from the functions on inputs that are very noisy, i.e.,
where the posterior probability of the function value is close to
uniform. Thus, a more practically motivated way of quantifying closeness
to SP is by considering noise sensitivity and stability.

Define the \emph{strong noise sensitivity} of a function $f$ to be
\[
\NS_{\delta}^{*}[f]\dfn\Pr\left(f(X^{n})\neq\sgn T_{\rho}f(Y^{n})\right),
\]
and the associated \emph{strong stability} as 
\[
\Stab_{\rho}^{*}[f]\dfn\E\left(f(X^{n})\cdot\sgn T_{\rho}f(Y^{n})\right).
\]
Of course, just as for the regular noise sensitivity and stability,
we have the trivial connection 
\[
\Stab_{\rho}^{*}[f]=1-2\NS_{\frac{1-\rho}{2}}^{*}[f].
\]
We can also express the strong stability in terms of the noise operator:
\begin{align*}
\Stab_{\rho}^{*}[f] & =\E\left(\E\left(f(X^{n})\cdot\sgn T_{\rho}f(Y^{n})\mid Y^{n}\right)\right)\\
 & =\E\left(T_{\rho}f(Y^{n})\cdot\sgn T_{\rho}f(Y^{n})\right)\\
 & =\E\left|T_{\rho}f(Y^{n})\right|\\
 & =\left\Vert T_{\rho}f\right\Vert _{1}.
\end{align*}
Thus the $1$-norm of $T_{\rho}f$ can be interpreted in terms of
the error probability associated with the optimal predictor for $f$.
Since the optimal predictor $\sgn T_{\rho}f$ can only do better than
$f$ itself, we immediately have: 
\begin{prop}
\label{prop: SP charctarization using stability and noise operator}For
any function $f$ and any $\rho\in[0,1]$ 
\[
\left\Vert T_{\sqrt{\rho}}f\right\Vert _{2}^{2}\leq\left\Vert T_{\rho}f\right\Vert _{1},
\]
with equality if and only if $f$ is $\rho$-SP. 
\end{prop}
The strong stability can also be upper bounded by a regular stability
expression. 
\begin{prop}
$\Stab_{\rho}[f]\leq\Stab_{\rho}^{*}[f]\leq\sqrt{\Stab_{\rho^{2}}[f]}$. 
\end{prop}
\begin{proof}
Write 
\begin{align*}
\Stab_{\rho}^{*}[f] & =\left\langle T_{\rho}f,\sgn T_{\rho}f\right\rangle \\
 & \leq\left\Vert T_{\rho}f\right\Vert _{2}\cdot\left\Vert \sgn T_{\rho}f\right\Vert _{2}\\
 & =\sqrt{\langle T_{\rho}f,T_{\rho}f\rangle}\\
 & =\sqrt{\langle T_{\rho^{2}}f,f\rangle}\\
 & =\sqrt{\Stab_{\rho^{2}}[f]}.
\end{align*}
where the inequality is by Cauchy-Schwartz's inequality, the next
equality is since $\|\sgn T_{\rho}f\|_{2}=1$, and the following equality
is since $T_{\rho}f$ is a self-adjoint operator (this follows from
Plancherel's identity: $\langle T_{\rho}f,g\rangle=\sum_{S\subseteq[n]}\rho^{|S|}\hat{f}_{S}\hat{g}_{S}=\langle f,T_{\rho}g\rangle$). 
\end{proof}
An immediate consequence of the above is: 
\begin{cor}
The strong noise sensitivity satisfies: 
\[
\frac{1-\sqrt{\Stab_{\rho^{2}}[f]}}{1-\Stab_{\rho}[f]}\cdot\NS_{\delta}[f]\leq\NS_{\delta}^{*}[f]\leq\NS_{\delta}[f].
\]
\end{cor}
Note that this bound is tight for the characters (and again shows
that they are USP). We can easily derive the following weaker statements: 
\begin{cor}
For any $f$ 
\begin{equation}
\frac{\NS_{\delta}[f]}{2}\leq\NS_{\delta}^{*}[f]\leq\NS_{\delta}[f].\label{eq: simple bounds on NS - half}
\end{equation}
If $f$ is balanced, then 
\begin{equation}
\frac{\NS_{\delta}[f]}{1+\rho}\leq\NS_{\delta}^{*}[f]\leq\NS_{\delta}[f].\label{eq: simple bound on NS - 1 plus rho^2}
\end{equation}
\end{cor}
\begin{proof}
The bounds in (\ref{eq: simple bounds on NS - half}) follow from
$\Stab_{\rho^{2}}[f]\leq\Stab_{\rho}[f]$ and $\min_{t\in[0,1]}\frac{1-\sqrt{t}}{1-t}=1/2$.
The bounds in (\ref{eq: simple bound on NS - 1 plus rho^2}) follow
from
\begin{align*}
\frac{1-\sqrt{\Stab_{\rho^{2}}[f]}}{1-\Stab_{\rho}[f]} & \geq\frac{1-\sqrt{\Stab_{\rho^{2}}[f]}}{1-\Stab_{\rho^{2}}[f]}\\
 & =\frac{1}{1+\sqrt{\Stab_{\rho^{2}}[f]}}\\
 & \geq\frac{1}{1+\rho}.
\end{align*}
\end{proof}
We may obtain improved bounds for low correlation values:
\begin{prop}
\label{prop:Prediction gain bounds for low correlation} Suppose $W^{1}[f]>0$.
Then: 
\[
\max\left\{ 1,\frac{1}{\sqrt{2W^{1}[f]}}+O(\rho^{2})\right\} \leq\frac{\Stab_{\rho}^{*}[f]}{\Stab_{\rho}[f]}\leq\frac{1}{\sqrt{W^{1}[f]}}+O(\rho^{2}).
\]
\end{prop}
\begin{proof}
We have that 
\begin{align*}
\Stab_{\rho}^{*}[f] & =\E\left|T_{\rho}f(Y^{n})\right|\\
 & =\E\left|\sum_{i=1}^{n}\rho\hat{f}_{i}Y_{i}\right|+O(\rho^{2}).
\end{align*}
Khintchine's inequality \cite{haagerup1981best} then implies 
\[
\frac{1}{\sqrt{2}}\cdot\sqrt{W^{1}[f]}\cdot\rho+O(\rho^{2})\leq\Stab_{\rho}^{*}[f]\leq\sqrt{W^{1}[f]}\cdot\rho+O(\rho^{2}),
\]
and the result follows from \cite[Proposition 2.51]{Bool_book}
\[
\Stab_{\rho}[f]=W^{1}[f]\cdot\rho+O(\rho^{2}).
\]
\end{proof}
\begin{cor}
For any balanced LTF $W^{1}[f]\geq1/2$ \cite[Theorem 5.2]{Bool_book},
and so 
\[
\frac{\Stab_{\rho}^{*}[f]}{\Stab_{\rho}[f]}\leq\sqrt{2}+O(\rho^{2}).
\]
\end{cor}

\section{High Correlation Sufficient Conditions\label{sec:high-corr}}

In this section, we derive sufficient conditions on a function to
be SP using various arguments. All our conditions will be high correlation
ones, i.e., for $\rho_{0}$ larger than some threshold. To that end,
we will need a simple characterization of monotone SP functions. Recall
that $x^{n}$ is called a \textit{boundary point} of $f$ if the value
of $f(x^{n})$ can be flipped by flipping some single coordinate of
$x^{n}$. We further say that $x^{n}$ is a \textit{dominating boundary
point} of $f$ if $f(x^{n})=1$ (resp. $=-1$) and $f(y^{n})=-1$
(resp. $=1$) for any $y^{n}\prec x^{n}$ (resp. $x^{n}\prec y^{n}$). 

The following is a simple corollary to the fact that monotonicity
is preserved by $\sgn T_{\rho}$ (Proposition \ref{prop: Optimal predictor preserves OSM}).
\begin{prop}
\label{prop:dominating_boundary} A monotone function is $\rho$-SP
if and only if it is $\rho$-SP at all its dominating boundary points.
\end{prop}
We can now prove the following:
\begin{prop}
\label{thm: Sufficient - no flips}Any function is $\rho$-SP for
$\rho>2^{\nicefrac{(n-1)}{n}}-1$, and there is no better universal
guarantee. 
\end{prop}
\begin{proof}
This range corresponds to the values of the crossover probability
$\delta\in[0,1-2^{-\nicefrac{1}{n}})$ for which the probability no
bit was flipped $(1-\delta)^{n}$, is at least $1/2$. This bound
is achieved with equality by the OR function $\OR(x^{n})$. To see
this, note that the OR function is monotone and symmetric with a single
dominating boundary point $1^{n}$. For this point 
\[
T_{\rho}\OR(1^{n})=(1-\delta)^{n}\cdot1+\left[1-(1-\delta)^{n}\right]\cdot(-1)
\]
which is non-negative if and only if $\delta\in[0,1-2^{-\nicefrac{1}{n}}]$. 
\end{proof}
Specific properties of the function, may be used to obtain better
sufficient bounds in special cases. For example, suppose that the
\emph{sparsity} of $\hat{f}_{S}$ is $s$, i.e., 
\[
f(x^{n})=\sum_{S\in\mathcal{S}}\hat{f}_{S}\cdot x^{S}
\]
where $\mathcal{S}\subset2^{[n]}$ and $|\mathcal{S}|=s$. Then, an
application of the union bounds leads to 
\begin{align*}
\P\left[f(X^{n})=f(y^{n})\mid Y^{n}=y^{n}\right] & \geq\P\left[\bigcap_{S\in\mathcal{S}}X^{S}=y^{S}\mid Y^{n}=y^{n}\right]\\
 & \geq1-\sum_{S\in\mathcal{S}}\P\left[X^{S}\neq y^{S}\mid Y^{n}=y^{n}\right]\\
 & =1-\sum_{S\in\mathcal{S}}\frac{1-\rho^{|S|}}{2}.
\end{align*}
This probability will be larger than $1/2$ for all $y^{n}\in\{-1,1\}^{n}$
if $\rho$ is larger than the solution to 
\[
\sum_{S\in\mathcal{S}}\rho^{|S|}=s-1.
\]
Similar conditions can be derived for PTFs (\ref{eq: PTF}) of sparsity
$s$. 

The extremal property of the OR function noted above may ostensibly
be attributed to the fact that it is extremely unbalanced. However,
$x_{1}\cdot\OR(x_{2}^{n})$ is balanced, and Propositions \ref{prop:seperable functions}
and \ref{thm: Sufficient - no flips} imply that it is $\rho$-SP
for $\rho>2^{\nicefrac{(n-2)}{(n-1)}}-1=1-\frac{2\ln(2)}{n}+O(n^{-2})$.
The next proposition demonstrates that the statement in Proposition
\ref{thm: Sufficient - no flips} holds even if we restrict ourselves
to balanced LTFs. 
\begin{prop}
\label{prop: balanced but very not SP}Any balanced LTF $f$ is $\rho$-SP
for $\rho>1-\frac{2\ln(2)}{n}+O(n^{-2})$, and there is no better
universal guarantee. 
\end{prop}
\begin{proof}
Note that the above region is essentially the same as the one in Proposition
\ref{thm: Sufficient - no flips}, hence one direction is clear. We
need to show there exists a balanced function that is not $\rho$-SP
at any point outside this region. To that end, let us introduce the
\emph{enlightened dictator} function, defined for $n\geq3$ to be
\begin{equation}
\edic(x^{n})\dfn\sgn\left((n-2)x_{1}+\sum_{i=2}^{n}x_{i}\right).\label{eq: E-dict def}
\end{equation}
Evidently, $\edic(x^{n})$ is determined by the ``dictator'' $x_{1}$,
unless all the ''subjects'' $x_{2},\ldots,x_{n}$ disagree. It is
easy to verify that $\edic(x^{n})$ is a monotone, odd (and hence
balanced) function. This function is SP at $y^{n}=(-1,1^{n-1})$ if
and only if 
\begin{align}
\P(\edic(X^{n})=1\mid Y^{n}=y^{n}) & =(1-\delta)^{n}+\delta(1-\delta^{n-1})\geq1/2.\label{eq:edict_type1}
\end{align}
The second derivative of the left-hand side (l.h.s.) above is $n(n-1)((1-\delta)^{n-1}-\delta^{n-2})$,
which is non-negative for $\delta\in[0,1/2]$, hence the l.h.s. is
convex inside this interval. It is easy to check that equality in
(\ref{eq:edict_type1}) holds for $\delta=\ln(2)\cdot n^{-1}-O(n^{-2})$
and for $\delta=1/2$, hence by convexity $y^{n}$ is $\delta$-SP
if and only if $\delta<\ln(2)\cdot n^{-1}-O(n^{-2})$, or equivalently,
$\rho>1-2\ln(2)\cdot n^{-1}+O(n^{-2})$.\footnote{Using Proposition \ref{prop:dominating_boundary} it can be verified
that this the true range for which $\edic(\cdot)$ is SP.}
\end{proof}

\subsection{Bounded Degree and Spectral Norm}

Next, we provide a stronger statement that uses the \emph{Fourier-degree}
$\Deg(f)$ of the function $f$, i.e., the maximal degree of the characters
appearing in the Fourier representation of $f$. 
\begin{thm}
\label{thm: Sufficient - Berenstein}Any function $f$ is $\rho$-SP
for 
\[
\rho\geq1-\frac{1}{\Deg(f)\cdot\min\left\{ \Deg(f),\sum_{S\subseteq[n]}\left|\hat{f}_{S}\right|\right\} }.
\]
\end{thm}
\begin{proof}
Fix any $y^{n}$ and think of $T_{\rho}f(y^{n})$ as a polynomial
in $\rho$. Let $\rho_{0}$ be the largest root of this polynomial
in $[0,1]$ (if there is one, otherwise $\rho_{0}=0$). Since $T_{\rho}f(y^{n})$
equals $f(y^{n})\in\{1,-1\}$ for $\rho=1$, then by continuity, $f$
is $\rho$-SP at $y^{n}$ for any $\rho\geq\rho_{0}$. By the mean
value theorem 
\[
1=T_{1}f(y^{n})-T_{\rho_{0}}f(y^{n})=(1-\rho_{0})\left.\frac{\d}{\d\rho}T_{\rho}f(y^{n})\right|_{\rho=\tilde{\rho}}
\]
for some $\tilde{\rho}\in[\rho_{0},1]$, and so 
\begin{equation}
\rho_{0}\leq1-\frac{1}{\max_{\rho\in[0,1]}\left|\frac{\d}{\d\rho}T_{\rho}f(y^{n})\right|},\label{eq: bound on largest root}
\end{equation}
and so a bound on $\rho_{0}$ may be obtained by bounding the derivative.
To that end, recall that Markov brothers' inequality \cite[Theorem 1.1]{govil1999markov}
states that for any real polynomial $P(t)$ of degree $k$
\[
\max_{t\in[-1,1]}\left|\frac{\d}{\d t}P(t)\right|\leq k^{2}\cdot\max_{t\in[-1,1]}\left|P(t)\right|,
\]
and that Bernstein's inequality \cite[Theorem 1.2]{govil1999markov}
states that for any complex polynomial $Q(z)$ of degree $k$, 
\[
\max_{|z|\leq1}\left|\frac{\d Q(z)}{\d z}\right|\leq k\cdot\max_{|z|\leq1}\left|Q(z)\right|.
\]
The claim then follows from (\ref{eq: bound on largest root}) by
noting that the degree of $T_{\rho}f$ as a polynomial in $\rho$
equals the Fourier degree $\Deg(f)$ , and the bound
\begin{align*}
|T_{\rho}f(y^{n})| & =\left|\sum_{S\subseteq[n]}\rho^{|S|}\cdot\hat{f}_{S}\cdot x^{S}\right|\\
 & \leq\sum_{S\subseteq[n]}|\hat{f}_{S}|
\end{align*}
for any $\rho\in(0,1]$. 
\end{proof}
Theorem \ref{thm: Sufficient - Berenstein} significantly improves
on Theorem \ref{thm: Sufficient - no flips} whenever $\Deg(f)\ll\sqrt{n}$,
e.g., for $n$-dimensional functions $f$ that can be computed by
a decision tree of depth $k\ll n$, in which case $\Deg(f)\leq k$
\cite[Proposition 3.16]{Bool_book}. Functions with low \emph{spectral
norm $\sum_{S\subseteq[n]}|\hat{f}_{S}|$ }are discussed in \cite{shpilka2017structure}
and references therein.

\subsection{Friendly Neighbors}

Given a function $f$, we say that a point $x^{n}$ has a \emph{radius-$d$
friendly neighborhood} w.r.t. $f$ if there exists some $y^{n}$ of
distance at most $d$ that agrees with $x^{n}$, namely, where $\Hamd(x^{n},y^{n})\leq d$
and $f(x^{n})=f(y^{n})$.
\begin{prop}
\label{prop: distance 2 neighbour}Suppose $f$ is $\rho$-SP for
all $\rho>1-\varepsilon$, and $n>\max\{2\varepsilon^{-1},\gamma\}$
where $\gamma$ is a universal constant. Then each point in $\{-1,1\}^{n}$
has a radius-$2$ friendly neighborhood w.r.t. $f$. 
\end{prop}
\begin{proof}
Suppose toward contradiction that all the neighbors at Hamming distance
$1$ and $2$ from some $y^{n}$ disagree with it. This implies that
\begin{align*}
\P\left(f(X^{n})\neq f(Y^{n})\mid Y^{n}=y^{n}\right) & \geq\binom{n}{1}\delta(1-\delta)^{n-1}+\binom{n}{2}\delta^{2}(1-\delta)^{n-2}\\
 & =(1-\delta)^{n-2}n\delta\left((1-\delta)+\frac{(n-1)}{2}\delta\right).
\end{align*}
Choosing $\delta=\frac{\alpha}{n}$, and assuming that $n>\frac{2\alpha}{\varepsilon}$
so that we are in the SP region, yields 
\begin{align*}
\P\left(f(X^{n})\neq f(Y^{n})|Y^{n}=y^{n}\right) & \geq\left(1-\frac{\alpha}{n}\right)^{n-2}\alpha\left(1+\frac{\alpha}{2}-\frac{3\alpha}{2n}\right)\\
 & \geq\left(1-\frac{\alpha}{n}\right)^{n-2}\cdot\left(\alpha+\frac{\alpha^{2}}{2}\right)-O\left(\frac{1}{n}\right)\\
 & =e^{-\alpha}\cdot\left(\alpha+\frac{\alpha^{2}}{2}\right)-O\left(\frac{1}{n}\right).
\end{align*}
One can check that, e.g., for $\alpha=1$, $(\alpha+\frac{\alpha^{2}}{2})e^{-\alpha}>1/2$,
and so $f$ cannot be SP if $n$ is larger than some universal constant,
in contradiction. 
\end{proof}
Hence, for a function to be SP even slightly below the guaranteed
high correlation threshold of $\rho>1-\frac{2\ln(2)}{n}+O(n^{-2})$,
every point must admit a radius-$2$ friendly neighborhood. The OR
function, e.g., does not satisfy this property. Furthermore, this
result is tight: for the largest character $x^{[n]}=\prod_{i=1}^{n}x_{i}$,
which is USP, the distance-$1$ neighbors of each point do not agree
with it.

The following corollary, which is not directly related to self-predictability,
is obtained by combining Theorem \ref{thm: Sufficient - Berenstein}
and Proposition \ref{prop: distance 2 neighbour}. 
\begin{cor}
If $\Deg f<\sqrt{n/2}$ and $n$ is larger than a universal constant,
then each point in $\{-1,1\}^{n}$ has a radius-$2$ friendly neighborhood
w.r.t. $f$. 
\end{cor}

\section{Low Correlation Self Predicting (LCSP) Functions\label{sec:low-corr}}

In this section we discuss LCSP functions, i.e., functions that are
$\rho$-SP for any $\rho<\rho^{*}$ for some $\rho^{*}>0$. Note that
any USP function is trivially also LCSP, hence all our LCSP necessary
conditions will apply to USP functions verbatim.

\subsection{LCSP and Spectral Threshold Functions}

Let the \emph{minimal level} of a function $f$ be defined as 
\[
\Lev(f)\dfn\min\left\{ k\in[n]:W^{k}[f]>0\right\} ,
\]
and let 
\[
f_{\Lev}(x^{n})\dfn\sum_{S:|S|=\Lev(f)}\hat{f}_{S}x^{S}.
\]
We say that $f$ is \emph{weakly spectral threshold (WST)} if $f_{\Lev}(x^{n})\cdot f(x^{n})\geq0$
for all $x^{n}$, i.e., the sign of both functions agree whenever
$f_{\Lev}\neq0$. We say that $f$ is \emph{strongly spectral threshold
(SST)} if it is WST and $f_{\Lev}$ is never zero.

For an LTF $f=\sgn(a_{0}+\sum_{i=1}^{n}a_{i}x_{i})$, the Fourier
coefficients $(\hat{f}_{\phi},\hat{f_{1}},\ldots,\hat{f}_{n})$ are
called \emph{Chow parameters}, and, as is well-known \cite{chow1961characterization,tannenbaum1961establishment},
these parameters unambiguously determine the LTF. The \emph{Chow-parameters
problem} \cite{o2011chow} is to find coefficients $a_{0}^{n}$ defining
the LTF given the Chow parameters. It can be seen that in case of
balanced LTFs, SST functions are exactly the LTFs for which a solution
to the Chow-parameters problem is exactly the Chow parameters themselves. 
\begin{prop}
\label{prop: USP implies WST}SST implies LCSP. Conversely, LCSP implies
WST. 
\end{prop}
\begin{proof}
The optimal predictor for $f$ satisfies 
\begin{align*}
\sgn T_{\rho}f(x^{n}) & =\sgn\left(\rho^{\Lev(f)}\cdot\sum_{s:|S|\geq\Lev(f)}\rho^{|S|-\Lev(f)}\hat{f}_{S}x^{S}\right)\\
 & =\sgn\left(f_{\Lev}(x^{n})+O(\rho)\right).
\end{align*}
Thus, $\sgn T_{\rho}f(x^{n})=\sgn f_{\Lev}(x^{n})$ for any $\rho$
small enough whenever $f_{\Lev}(x^{n})\neq0$. If $f$ is SST $f_{\Lev}(x^{n})$
never vanishes, and hence $f(x^{n})=\sgn f_{\Lev}(x^{n})=\sgn T_{\rho}f(x^{n})$,
implying LCSP. Conversely, if $f$ is LCSP, then $f(x^{n})=\sgn T_{\rho}f(x^{n})=\sgn f_{\Lev}(x^{n})$
unless $f_{\Lev}$ vanishes, implying WST. 
\end{proof}
An immediate consequence of Proposition \ref{prop: USP implies WST}
is: 
\begin{cor}
\label{cor: unblanaced functions are SP only constants} An LCSP function
is either balanced or constant. 
\end{cor}
\begin{proof}
Suppose $f$ is LCSP and unbalanced. Then $\Lev[f]=0$ and $\hat{f}_{\phi}\neq0$,
and by Proposition \ref{prop: USP implies WST} it must be WST. Hence
$f=\sgn\hat{f}_{\phi}\in\{-1,1\}$ must be constant. 
\end{proof}
It is also interesting to note the following dichotomy:
\begin{cor}
Let $f$ be an LCSP function. Then either $W^{1}[f]=0$ or $W^{1}[f]\geq1/2$.
\end{cor}
\begin{proof}
If $0<W^{1}[f]<1/2$ then Proposition \ref{prop:Prediction gain bounds for low correlation}
implies that
\[
\frac{\Stab_{\rho}^{*}[f]}{\Stab_{\rho}[f]}>1
\]
for all sufficiently small $\rho$, and so $f$ cannot be LCSP. 
\end{proof}
This result resembles the claim that $W^{1}[f]\geq1/2$ for LTFs \cite[Theorem 5.2]{Bool_book}.
Note however that the above claim holds for LCSP functions that are
not LTFs but do have energy on the first level. Next, recall that
Proposition \ref{prop: SP charctarization using stability and noise operator}
states that a function is $\rho$-SP if and only if $\Vert T_{\rho}f\Vert_{1}=\Stab_{\rho}^{*}[f]=\Stab_{\rho}[f]=\Vert T_{\sqrt{\rho}}f\Vert_{2}^{2}$.
A similar property holds for $f_{\Lev}$ if the function is LCSP. 
\begin{cor}
If $f$ is LCSP then $\Vert f_{\Lev}\Vert_{1}=\Vert f_{\Lev}\Vert_{2}^{2}$. 
\end{cor}
\begin{proof}
$f$ must be WST by Proposition \ref{prop: USP implies WST}, and
so Plancherel's identity implies that 
\begin{align*}
\E\left|f_{\Lev}(X^{n})\right| & =\E\left(f_{\Lev}(X^{n})\cdot f(X^{n})\right)\\
 & =\langle f_{\Lev},f\rangle\\
 & =\sum_{S:|S|=\Lev[f]}\hat{f}_{S}^{2}\\
 & =\E\left(f_{\Lev}^{2}(X^{n})\right).
\end{align*}
\end{proof}
The following two examples show that the distinction between WST and
SST in Proposition \ref{prop: USP implies WST} is necessary. 
\begin{example}[LCSP does not imply SST]
Consider the balanced LTF with $n=4$ and coefficients $a_{1}^{4}=(2,1,1,1)$.
This is a Majority function with a tie-breaking input. It can be verified
by direct computation that this function is USP, hence also LCSP.
However, its level-$1$ Fourier coefficients are $(\frac{3}{4},\frac{1}{4},\frac{1}{4},\frac{1}{4})$.
Hence, while it is clearly WST, it is not SST as there are $2$ inputs
for which $f_{\Lev}(x^{n})=0$. 
\end{example}
\begin{example}[WST does not imply LCSP]
The balanced LTF with $n=9$ and coefficients $a_{1}^{9}=(1,5,16,19,25,58,68,91,94)$
can be verified to be WST, but not LCSP. It is $\rho$-SP only for
$\rho>0.577$. This example was found by analyzing LTFs with randomly
drawn coefficients.
\end{example}
The following example shows that the SST property is limited to the
low-correlation regime only. 
\begin{example}[SST does not imply USP]
The LTF of Example \ref{exa: LTF example} is SST, but as was shown
there, is not USP. Thus, while an SST is always LCSP, it is not necessarily
USP. 
\end{example}
We note in passing that there are SST and WST functions outside Majority
that are USP. 
\begin{example}
The LTF in Example \ref{exa:USP not majority} is SST and USP, while
the balanced LTF with $n=9$ and coefficients $a_{1}^{9}=(1,1,1,3,3,3,5,5,7)$
is WST and USP ($f_{\Lev}=0$ for $30$ inputs), but not SST. 
\end{example}
Next, using Proposition \ref{prop: USP implies WST}, we can show
that the largest coefficients of an LCSP LTF cannot be too distinct. 
\begin{prop}
\label{prop:first to second ration for LTF}Let $f$ be an LTF that
depends on all its $n$ variables. Let $a$ and $b$ be its first
and second largest coefficients in absolute values, respectively,
in some representation of $f$. If $f$ is LCSP then $\left|\frac{a}{b}\right|<\sqrt{2n\ln(2n)}+1$.
\end{prop}
\begin{proof}
Assume without loss of generality that $a_{1}\geq a_{2}\geq\cdots\geq a_{n}>0$.
Recall also that by Corollary \ref{cor: unblanaced functions are SP only constants}
we know that $a_{0}=0$. Since $f$ is monotone, its level-$1$ Fourier
coefficients equal influences \cite[Proposition 2.21]{Bool_book},
i.e.,
\begin{align}
\hat{f}_{k} & =\Inf_{k}[f]\label{eq: Fourier is influence for monotone}\\
 & \dfn\P\left(f(X^{n})\neq f(X_{1}^{k-1},-X_{k},X_{k+1}^{n})\right)\nonumber \\
 & =\P\left(\left|\sum_{i\neq k}a_{i}X_{i}\right|<a_{k}\right).\label{eq: Fourier level 1 with probability}
\end{align}
Assume without loss of generality that $a_{2}=1$, and write $a\dfn a_{1}$.
For brevity, also write $Z\dfn\sum_{i=3}^{n}a_{i}X_{i}$ and $X\dfn X_{1}$.
Then, from the symmetry of $Z$, 
\begin{align*}
\hat{f}_{1} & =\P(|X+Z|\leq a)\\
 & =\P(|1+Z|\leq a)\\
 & \geq\P(|Z|<a-1),
\end{align*}
and 
\begin{align*}
\hat{f}_{2} & =\P(|aX+Z|\leq1)\\
 & \leq\P(a-1\leq|Z|\leq a+1)\\
 & \leq\P(|Z|\geq a-1).
\end{align*}
Hence, 
\[
\frac{\hat{f}_{1}}{\hat{f}_{2}}\geq\frac{1-\P(|Z|\geq a-1)}{\P(|Z|\geq a-1)}.
\]
Since $|a_{i}|\leq1$ for $3\leq i\leq n$, and assuming toward contradiction
that $a>\sqrt{(2n-2)\ln2n}+1$, Hoeffding's inequality implies that
\[
\P(|Z|\geq a-1)<1/n,
\]
and so $\hat{f}_{1}/\hat{f}_{2}>n-1$. Noting that $a_{i}\geq a_{j}$
implies $\hat{f}_{i}\geq\hat{f}_{j}$, we also have that $\hat{f}_{1}/\hat{f}_{i}\geq n-1+\varepsilon$
for any $i>1$, for $\varepsilon>0$ small enough. By Proposition
\ref{prop: USP implies WST}, $f$ is WST, i.e., $f(x^{n})=\sgn\sum_{i=1}^{n}\hat{f}_{i}x_{i}$
whenever the right-hand side (r.h.s.) is nonzero. This representation
and the bounds on the ratios $\hat{f}_{1}/\hat{f}_{i}$ from above
imply that $f(x^{n})=x_{1}$ must hold. This, however, contradicts
the assumption that $f$ depends on all the variables. 
\end{proof}
For example, the enlightened dictator function $\edic(\cdot)$ (\ref{eq: E-dict def})
has first-to-second coefficient ratio of $n-2$, and thus cannot be
LCSP. It should be noted however, that $\edic(\cdot)$ can also be
written as an LTF with coefficients $\edic(\cdot)=(\sqrt{n},1,c,c,\ldots,c)$
where $c=\frac{\sqrt{n}-1+\varepsilon}{n-2}$ for some $\varepsilon>0$.
When given in this form, Proposition \ref{prop:first to second ration for LTF}
is incapable of ruling it out from being SP. Nonetheless, it is easy
to verify that LTFs of coefficients $(c,1,1,...,1)$ for $c<n-2$
and $c=\Omega(n)$, must have $a_{2}=a_{3}\cdots=a_{n}$ in any valid
representation, and thus the first-to-second-coefficient ratio is
always $\Omega(n)$.

\subsection{LTF Approximation}

The WST condition can be leveraged to show that a LCSP function can
typically be well approximated by an LTF. Specifically: 
\begin{thm}
\label{thm:LCSP_LTF_approx} An LCSP $f$ is $\sqrt{\frac{2}{\pi n_{f}}}$-close
to an LTF, where $n_{f}\dfn|\{i\in[n]:\hat{f}_{i}\neq0\}|$. 
\end{thm}
\begin{cor}
A monotone LCSP function that depends on all its coordinates is $\sqrt{\frac{2}{\pi n}}$-close
to an LTF. 
\end{cor}
To prove Theorem \ref{thm:LCSP_LTF_approx} we first establish the
following technical lemma. We state it in a slightly more general
form than we actually need. 
\begin{lem}
\label{lem:sperner} Let $a^{n}\in\mathbb{R}^{n}$ be a vector of
nonzero coefficients. Then for any $b\in\mathbb{R}$ 
\[
\P\left(\left|\sum_{i=1}^{n}a_{i}X_{i}-b\right|<\min_{k\in[n]}|a_{k}|\right)\leq2^{-n}\binom{n}{\lfloor\nicefrac{n}{2}\rfloor}\leq\sqrt{\frac{2}{\pi n}}.
\]
\end{lem}
\begin{proof}
Write $a=\min|a_{k}|$ and let 
\[
\mathcal{A}\dfn\left\{ x^{n}\in\{-1,1\}^{n}:\left|\sum_{i=1}^{n}a_{i}x_{i}-b\right|<a\right\} .
\]
It is easy to see that $\mathcal{A}$ forms an antichain w.r.t. the
partial order $\preceq$ on $\{-1,1\}^{n}$, i.e., that there are
no two distinct $x^{n},y^{n}\in\mathcal{A}$ such that $x^{n}\preceq y^{n}$.
This holds simply since for such a pair it must hold that 
\[
\left|\sum_{i=1}^{n}a_{i}y_{i}-\sum_{i=1}^{n}a_{i}x_{i}\right|\geq2a.
\]
An antichain w.r.t. $\preceq$ is called a \emph{Sperner family},
and \emph{Sperner's theorem} \cite[Maximal Antichains, Corollary 2]{alon2004probabilistic}
shows that 
\[
|\mathcal{A}|\leq\binom{n}{\lfloor\nicefrac{n}{2}\rfloor}
\]
concluding the proof. 
\end{proof}
\begin{proof}[Proof of Theorem \ref{thm:LCSP_LTF_approx}]
Assume $\Lev[f]=1$ (trivial otherwise), and define $g(x^{n})=\sgn(\sum_{i=1}^{n}\hat{f}_{i}x_{i})$.
Let $\mathcal{A}\dfn\{x^{n}\in\{-1,1\}^{n}:g(x^{n})=0\}$. Using Lemma
\ref{lem:sperner}, we have that 
\begin{align*}
\Pr(X^{n}\in\mathcal{A}) & \leq\sqrt{\frac{2}{\pi n_{f}}}.
\end{align*}
Since $f$ is LCSP then by Proposition \ref{prop: USP implies WST}
is it also WST, and hence $f(x^{n})=g(x^{n})$ for all $x^{n}\not\in\mathcal{A}$.
By slightly perturbing the coefficients of $g$, one can clearly obtain
a ``legal'' LTF $\tilde{g}$ that takes values only in $\{-1,1\}$
and still agrees with $f$ for all $x^{n}\not\in\mathcal{A}$. The
distance between $f$ and $\tilde{g}$ is therefore at most $|\mathcal{A}|/2^{n}$. 
\end{proof}

\subsection{Chow Distance}

The \emph{Chow distance} between two Boolean functions $f$ and $g$
is defined as 
\[
d_{\Chow}(f,g)\dfn\left(\sum_{i\in[n]}\left(\hat{f}_{i}-\hat{g}_{i}\right)^{2}\right)^{1/2}.
\]
It was shown in \cite[Prop. 1.5, Th. 1.6]{o2011chow} that for any
$f$ and $g$ 
\[
\frac{1}{4}d_{\mathrm{\Chow}}^{2}(f,g)\le\Dist(f,g)\leq\tilde{O}\left(\frac{1}{\sqrt{-\log d_{\Chow}(f,g)}}\right),
\]
where for $q<1$, $\tilde{O}(q)$ means $O(q\cdot\log^{c}(1/q))$
for some absolute constant $c$. 

For LCSP LTF functions, the upper bound can be generally improved.
We will state our result for the case where one of the functions is
SST, though it can be somewhat cumbersomely extended to the case where
none of them is. Let $\Gap[f]$ be the minimal positive value of $\sum_{i=1}^{n}\hat{f}_{i}x_{i}$
over the Hamming cube (with $\Gap[f]=0$ if all the $\hat{f}_{i}$'s
are zero).
\begin{thm}
\label{thm:lcsp_separation} Let $f$ and $g$ be two balanced LCSP
functions that depend on all $n$ variables, and assume that $f$
is SST. Then 
\[
\Dist(f,g)\leq\frac{d_{\mathrm{\Chow}}^{2}(f,g)}{2\Gap[f]}.
\]
\end{thm}
\begin{proof}[Proof of Theorem \ref{thm:lcsp_separation}]
Let 
\[
\mathcal{B}\dfn\left\{ x^{n}\in\{-1,1\}^{n}:f(x^{n})\neq g(x^{n})\right\} .
\]
Then,
\begin{align}
d_{\Chow}^{2}(f,g) & =\E\left((f(X^{n})-g(X^{n}))\cdot\sum_{i\in[n]}\left(\hat{f}_{i}-\hat{g}_{i}\right)X_{i}\right)\label{eq:lcsp_dist1}\\
 & =2\E\left(\left|\sum_{i\in[n]}\left(\hat{f}_{i}-\hat{g}_{i}\right)X_{i}\right|\cdot\I(X^{n}\in\mathcal{B})\right)\label{eq:lcsp_dist2}\\
 & \geq2\Gap[f]\cdot\P\left(X^{n}\in\mathcal{B}\right),\label{eq:lcsp_dist3}
\end{align}
where (\ref{eq:lcsp_dist1}) follows from linearity of expectation
and the definition of the Fourier coefficients, (\ref{eq:lcsp_dist2})
holds since both $f$ and $g$ are WST by virtue of Proposition \ref{prop: USP implies WST}
and so for all $x^{n}\in\mathcal{B}$, $|f(X^{n})-g(X^{n})|=2$ and
$\sgn[\sum_{i\in[n]}(\hat{f}_{i}-\hat{g}_{i})X_{i}]=\sgn[f(X^{n})-g(X^{n})]$.
Finally, (\ref{eq:lcsp_dist3}) holds by noting that $f$ is SST and
$g$ is WST. Thus, whenever $f(x^{n})>g(x^{n})$ then $\sum_{i\in[n]}\hat{f}_{i}X_{i}>\Gap[f]$
and $\sum_{i\in[n]}\hat{g}_{i}X_{i}\leq0$ (and similarly for $f(x^{n})<g(x^{n})$). 
\end{proof}
Equations (\ref{eq: Fourier is influence for monotone})-(\ref{eq: Fourier level 1 with probability})
and Lemma \ref{lem:sperner} imply that $\Gap[\maj]\leq\sqrt{\frac{2}{\pi n}}$.
Since Majority is SST, we have:
\begin{cor}
For odd $n$ and any LCSP function $g$, 
\[
\frac{1}{4}\cdot d_{\Chow}^{2}(\maj,g)\leq\Dist(\maj,g)\leq\sqrt{\frac{\pi n}{8}}\cdot d_{\Chow}^{2}(\maj,g).
\]
\end{cor}

\section{Stability-based Conditions\label{sec:stab}}

In this section we provide simple necessary conditions for a function
to be $\rho$-SP, in terms of its stability and Fourier coefficients.
\begin{prop}
\label{prop: necessary condition basic}If $f$ is $\rho$-SP then
\begin{align*}
\Stab_{\rho}[f] & \geq\max_{S\subseteq[n]}\rho^{|S|}|\hat{f}_{S}|.
\end{align*}
\end{prop}
\begin{proof}
If $f$ is $\rho$-SP, then $\Stab_{\rho}[f]=\Stab_{\rho}^{*}[f]$.
Letting $T\subseteq[n]$, the strong stability can be lower bounded
as follows: 
\begin{align*}
\Stab_{\rho}^{*}[f] & =\E\left|\sum_{S\subseteq[n]}\rho^{|S|}\cdot\hat{f}_{S}\cdot Y^{S}\right|\\
 & =\E\left(\left|\sum_{S\subseteq[n]}\rho^{|S|}\cdot\hat{f}_{S}\cdot Y^{S}\right|\cdot\left|Y^{T}\right|\right)\\
 & =\E\left(\left|\sum_{S\subseteq[n]}\rho^{|S|}\cdot\hat{f}_{S}\cdot Y^{S}\cdot Y^{T}\right|\right)\\
 & \geq\left|\E\left(\sum_{S\subseteq[n]}\rho^{|S|}\cdot\hat{f}_{S}\cdot Y^{S}\cdot Y^{T}\right)\right|\\
 & =|\rho^{|T|}\cdot\hat{f}_{T}|.
\end{align*}
The proof is completed by optimizing over $T$. 
\end{proof}
\begin{example}
When $f$ is the OR function, we have 
\[
\max_{S\subseteq[n]}\rho^{|S|}|\hat{f}_{S}|=|\hat{f}_{\phi}|=1-2^{1-n}.
\]
It is easy to verify that 
\[
\P\left(f(X^{n})=f(Y^{n})\right)=1-2^{1-n}\cdot\left(1-\left(1-\delta\right)^{n}\right),
\]
and using $\rho=1-2\delta$ 
\begin{align*}
\Stab_{\rho}[f] & =2\cdot\P\left(f(X^{n})=f(Y^{n})\right)-1\\
 & =1-2^{2-n}\cdot\left(1-\left(\frac{1+\rho}{2}\right)^{n}\right).
\end{align*}
Then, $\OR$ is $\rho$-SP only when $\Stab_{\rho}[\OR]\geq1-2^{1-n}$,
which can be seen to be equivalent to $\rho\geq2^{(\nicefrac{n-1)}{n}}-1$.
This is the same result that can be obtained by direct computation
(see Proposition \ref{thm: Sufficient - no flips}), and so the bound
of Proposition \ref{prop: necessary condition basic} is tight in
this case. Furthermore, we may deduce again the result of Corollary
\ref{cor: unblanaced functions are SP only constants}: 
\end{example}
\begin{cor}
An LCSP function is either balanced or constant. 
\end{cor}
\begin{proof}
If $f$ is $\rho$-SP then 
\[
\Stab_{\rho}[f]=\sum_{S\subseteq[n]}\rho^{|S|}|\hat{f}_{S}|^{2}\geq|\hat{f}_{\phi}|.
\]
As $\rho\downarrow0$, this bound implies that $|\hat{f}_{\phi}|^{2}\geq|\hat{f}_{\phi}|$,
and as $|\hat{f}_{\phi}|\leq1$, this is only possible when either
$\hat{f}_{\phi}=0$ or $|\hat{f}_{\phi}|=1$. 
\end{proof}
More generally, we have the following: 
\begin{cor}
If $f$ is LCSP then 
\[
W^{\Lev[f]}[f]\geq\max_{S\subseteq[n]:\;|S|=\Lev(f)}|\hat{f}_{S}|.
\]
Specifically, if $f$ is also monotone, this bound reads 
\[
W^{1}[f]\geq\max_{i\in[n]}\hat{f}_{i}=\max_{i\in[n]}\Inf_{i}[f],
\]
where the r.h.s. is the so-called \emph{maximal influence} of $f$. 
\end{cor}
When $\Deg(f)<n$, another bound of the form of Proposition \ref{prop: necessary condition basic}
can be derived using the following implication of hypercontractivity
\cite{bonami1970etude,gross1975logarithmic}: When $f:\{-1,1\}^{n}\to\mathbb{R}$
has $\Deg(f)=k$ then $\left\Vert f\right\Vert _{2}\leq e^{k}\cdot\left\Vert f\right\Vert _{1}$
\cite[Theorem 9.22]{Bool_book}. 
\begin{prop}
\label{prop: necessary condition hyper}If $f$ is $\rho$-SP and
$\Deg(f)=k$ then 
\[
\Stab_{\rho}[f]\geq e^{-k}\cdot\sqrt{\Stab_{\rho^{2}}[f]}.
\]
\end{prop}
\begin{proof}
As in the proof of Proposition \ref{prop: necessary condition basic},
we lower bound 
\begin{align}
\E\left|\sum_{S\subseteq[n]}\rho^{|S|}\cdot\hat{f}_{S}\cdot Y^{S}\right| & =\left\Vert T_{\rho}f\right\Vert _{1}\nonumber \\
 & \geq e^{-k}\cdot\left\Vert T_{\rho}f\right\Vert _{2}\label{eq: hypercontractivity stability bound a}\\
 & =e^{-k}\cdot\sqrt{\langle T_{\rho}f,T_{\rho}f\rangle}\nonumber \\
 & =e^{-k}\cdot\sqrt{\langle T_{\rho^{2}}f,f\rangle}\label{eq: hypercontractivity stability bound b}\\
 & =e^{-k}\cdot\sqrt{\Stab_{\rho^{2}}[f]}\nonumber 
\end{align}
where (\ref{eq: hypercontractivity stability bound a}) is since $\Deg(f)=\Deg(T_{\rho}f)=k$,
and (\ref{eq: hypercontractivity stability bound b}) is since $T_{\rho}f$
is a self-adjoint operator. 
\end{proof}
The last proof implies for a degree $k$, $\rho$-SP function $f$
\[
e^{-k}\cdot\sqrt{\Stab_{\rho^{2}}[f]}\leq\Stab_{\rho}[f]\leq\sqrt{\Stab_{\rho^{2}}[f]}.
\]
It can be observed that even for a given degree $k$, neither of the
bounds in Propositions \ref{prop: necessary condition basic} and
\ref{prop: necessary condition hyper} subsumes the other.

\section{Sharp Threshold at High Correlation \label{sec:Sharp-Threshold}}

As we have seen, all functions are $\rho$-SP when $\rho>1-\frac{2\ln2}{n}+O(n^{-2})$.
In this section, we show that when the correlation is reduced ever
so slightly to $\rho\approx1-\frac{2}{n}$, the fraction of SP functions
becomes double-exponentially small.
\begin{thm}
\label{thm: sharp threshold delta=00003D1/n}For any $\alpha>1$,
the fraction of $\rho$-SP functions for $\rho=1-\frac{2\alpha}{n}$
is at most $\exp(-2^{n\cdot E(\alpha)+o(n)})$, where 
\[
E(\alpha)\dfn\min\left\{ \frac{1}{2},\binent\left(\frac{\alpha-1}{2\alpha}\right)\right\} 
\]
and $\binent(t)\dfn-t\log(t)-(1-t)\cdot\log(1-t)$ is the binary entropy
function. 
\end{thm}
The fact that $\rho$-SP functions are rare is not limited to the
$\rho=1-O(\frac{1}{n})$ regime, yet a different technique is needed
in order to establish this in other regimes. We next demonstrate how
a similar phenomenon holds in a high correlation regime where $\rho$
is fixed. Let $\eta_{\delta}$ be the minimal $\eta>0$ such that
\[
\frac{1}{2}\log\frac{1}{\delta^{2}+(1-\delta)^{2}}<\min\left\{ \log\frac{1}{1-\delta},\bindiv(\eta||\delta)\right\} 
\]
holds, where $\bindiv(p||q)\dfn p\log\frac{p}{q}+(1-p)\log\frac{1-p}{1-q}$
is the binary divergence function. It can be verified that $\eta_{\delta}<1/4$
for any $\delta<\delta_{\max}\approx0.0974$. 
\begin{thm}
\label{thm: sharp threshold large delta}For any $\delta\in(0,\delta_{\max})$,
the fraction of $\rho$-SP functions for $\rho=1-2\delta$ is at most
$\exp\left(-2^{n[1-\binent(2\eta_{\delta})]-o(n)}\right)$. 
\end{thm}
We begin with the proof of Theorem \ref{thm: sharp threshold delta=00003D1/n}.
\begin{proof}[Proof of Theorem \ref{thm: sharp threshold delta=00003D1/n}]
In this proof we find it more convenient to work with a $\delta$
and $\{0,1\}$ convention. We begin by deriving a sufficient condition
for a function to be non-$\rho$-SP at any fixed $y^{n}$. This condition
depends only on local values of the function, up to a Hamming distance
of $\log n$ from $y^{n}$, and is tailored to the regime of $\delta=\Theta(1/n)$.
Specifically, we show that for a random choice of function, the probability
that our condition is satisfied decays exponentially with $n$, and
we derive an upper bound on the associated exponent. Then, since the
resulting exponent is smaller than $1$, we conclude that the expected
number of non-$\rho$-SP points for a random function is exponentially
large. This fact in itself, however, is not sufficient since there
are statistical dependencies between different points in the Hamming
cube. Nonetheless, Janson\textquoteright s theorem \cite[Theorem 8.1.1]{alon2004probabilistic}
along with the aforementioned ``locality'' of the sufficient condition
allow us to prove that the probability that all points in the Hamming
cube are $\rho$-SP is only double-exponentially small. 

We proceed to prove the local condition for non-$\rho$-SP-ness. To
that end, let us denote the \emph{shell of radius $d$ around $x^{n}\in\{0,1\}^{n}$}
by 
\[
\mathcal{S}(x^{n},d)\dfn\left\{ \tilde{x}^{n}:\Hamd(x^{n},\tilde{x}^{n})=d\right\} .
\]
For any function $f:\{0,1\}^{n}\to\{0,1\}$, let the $d$\emph{-shell
bias of $f$ }be 
\[
\beta_{d,f}(x^{n})\dfn\frac{1}{\left|\mathcal{S}(x^{n},d)\right|}\sum_{\tilde{x}^{n}\in\mathcal{S}(x^{n},d)}f(\tilde{x}^{n}).
\]
Fix $\eta>0$  and some $y^{n}$. Without loss of generality, we assume
below that $f(y^{n})=0$. Define the set of functions 
\[
\mathcal{B}_{\eta}(y^{n},1)=\left\{ f:\beta_{1,f}(y^{n})\geq1-\eta\right\} ,
\]
and for $2\leq d\leq\ell$, the sets 
\[
\mathcal{B}_{\eta}(y^{n},d)=\left\{ f:\beta_{d,f}(y^{n})\geq1/2\right\} ,
\]
where $\ell\geq3$. We say that $y^{n}$ is \emph{bad} for $f$ if
$f\in\mathcal{B}_{\eta}(y^{n})$, where 
\[
\mathcal{B}_{\eta}(y^{n})\dfn\bigcap_{d=1}^{\ell}\mathcal{B}_{\eta}(y^{n},d).
\]
Now, for any $n>\ell$, setting $\delta=\frac{\alpha}{n}$, any $f\in\mathcal{B}_{\eta}(y^{n})$
satisfies: 
\begin{align}
 & \P\left(f(X^{n})\neq f(Y^{n})\mid Y^{n}=y^{n}\right)\nonumber \\
 & \geq\sum_{d=1}^{\ell}\beta_{d,f}(y^{n})\binom{n}{d}\delta^{d}(1-\delta)^{n-d}\nonumber \\
 & \geq(1-\delta)^{n}\cdot\sum_{d=1}^{\ell}\beta_{d,f}(y^{n})\binom{n}{d}\delta^{d}\nonumber \\
 & \geq\left(1-\frac{\alpha}{n}\right){}^{n}\cdot\left(1-\frac{\ell}{n}\right){}^{\ell}\cdot\left(\sum_{d=1}^{\ell}\frac{\beta_{d,f}(y^{n})}{d!}\cdot\alpha^{d}\right)\nonumber \\
 & \geq\left(1-\frac{\alpha}{n}\right){}^{n}\cdot\left(1-\frac{\ell}{n}\right){}^{\ell}\cdot\left((1-\eta)\cdot\alpha+\frac{1}{2}\sum_{d=2}^{\ell}\frac{\alpha^{d}}{d!}\right)\nonumber \\
 & =\left(1-\frac{\alpha}{n}\right){}^{n}\cdot\left(1-\frac{\ell}{n}\right){}^{\ell}\cdot\left((1-\eta)\cdot\alpha+\frac{1}{2}\left(e^{\alpha}-1-\alpha-\sum_{d=\ell+1}^{\infty}\frac{\alpha^{d}}{d!}\right)\right).\label{eq: lower bound on conditional prob}
\end{align}
Taking $\ell$ to be $\Omega(1)$ and $o(n)$, say $\ell=\log n$,
(\ref{eq: lower bound on conditional prob}) tends to 
\[
\frac{1}{2}+\left(\frac{1}{2}-\eta\right)\alpha e^{-\alpha}-\frac{1}{2}e^{-\alpha}
\]
as $n\to\infty$. Let 
\[
\eta_{\alpha}\dfn\frac{\alpha-1}{2\alpha}.
\]
Clearly, $\eta_{\alpha}$ is monotonically increasing for $\alpha>0$,
where $\lim_{\alpha\downarrow1}\eta_{\alpha}=0$, and $\lim_{\alpha\uparrow\infty}\eta_{\alpha}=1/2$.
Setting $\eta\in(0,\eta_{\alpha})$ guarantees that (\ref{eq: lower bound on conditional prob})
is larger than $1/2$ for all large enough $n$. Hence, for such a
choice, 
\[
\P\left(f(X^{n})\neq f(Y^{n})\mid Y=y^{n}\right)>1/2,
\]
and so
\[
\left\{ f\in\mathcal{B}_{\eta}(y^{n})\right\} \subseteq\left\{ f\text{ is not }\text{\ensuremath{\rho}}\text{-SP at }y^{n}\right\} .
\]
Let us now choose $f$ uniformly at random over all Boolean functions
on $\{0,1\}^{n}$, and lower bound the probability that $f$ is $\rho$-SP
at $y^{n}$. To that end, note that Chernoff's bound implies that
\begin{equation}
\P\left(\beta_{1,f}(y^{n})\geq1-\eta\right)=2^{-n\left(1-\binent(\eta)\right)+o(n)},\label{eq: M1 large value asymptotic probability}
\end{equation}
and symmetry implies that
\[
\P\left(\beta_{d,f}(y^{n})\geq1/2\right)\geq1/2,
\]
for $2\leq d\leq\ell=\log n$. By independence,
\begin{align*}
\P\left(f\in\mathcal{B}_{\eta}(y^{n})\right) & =\prod_{d=1}^{\log n}\P\left(f\in\mathcal{B}_{\eta}(y^{n},d)\right)\\
 & =2^{-n\left(1-\binent(\eta)\right)+o(n)}\cdot2^{-(\log n-1)}\\
 & =2^{-n\left(1-\binent(\eta)\right)+o(n)},
\end{align*}
and so
\[
\P\left(f\text{ is not }\text{\ensuremath{\rho}}\text{-SP at }y^{n}\right)\geq2^{-n\left(1-\binent(\eta)\right)+o(n)}.
\]
This completes the proof of the local bound.

We now proceed to the global behavior of the number of non-$\rho$-SP
points. Let us first upper bound the probability $\P(f\in\mathcal{E})$
where 
\[
\mathcal{E}\dfn\bigcap_{y^{n}\in\{0,1\}^{n}}\mathcal{B}_{\eta}^{c}(y^{n}).
\]
This in turn will serve as an upper bound for the probability that
the function we draw is $\rho$-SP for the aforementioned $\rho$.
To that end, note that if $f$ is $\rho$-SP for $\rho=1-2\alpha\cdot n^{-1}$
then it must be that $f$ has no bad inputs, i.e., $f\in\mathcal{E}$.
Furthermore, note that the expected number of ``bad'' inputs is
given by
\[
\mu\dfn2^{n}\cdot\P\left(f\in\mathcal{B}_{\eta}(y^{n})\right)=2^{n\binent(\eta)+o(n)}.
\]
If the number of bad inputs had been Poisson distributed with mean
$\mu$, then 
\[
\P\left(f\in\mathcal{E}\right)=e^{-\mu}=\exp\left(-2^{n\cdot\binent(\eta)+o(n)}\right).
\]
However, the events $\mathcal{B}_{\eta}(x^{n})$ and $\mathcal{B}_{\eta}(y^{n})$
are dependent whenever $\Hamd(x^{n},y^{n})\leq2\ell$. Nonetheless,
Janson's correction \cite[Theorem 8.1.1]{alon2004probabilistic} implies
that 
\[
\P\left(f\in\mathcal{E}\right)\leq e^{-\mu+\frac{\Delta}{2}},
\]
where $\Delta$ is a correction term that depends on joint probability
of dependent bad events $\P(f\in\mathcal{B}_{\eta}(x^{n})\cap\mathcal{B}_{\eta}(y^{n}))$.
We next show that $\Delta\to0$ as $n\to\infty$ exponentially fast,
as long as $\eta\in(0,\binent^{-1}(1/2))$. Once this is established,
one can set $\eta=\min\{\eta_{\alpha},\binent^{-1}(1/2)\}$ to obtain
\[
\P\left(f\in\mathcal{E}\right)\leq\exp\left(-2^{n\cdot\binent(\eta)+o(n)}\right),
\]
and the theorem follows. 

To complete the proof, it remains to show that $\Delta\to0$ exponentially
fast. Let us denote $x^{n}\sim y^{n}$ whenever the events $\{f\in\mathcal{B}_{\eta}(x^{n})\}$
and $\{f\in\mathcal{B}_{\eta}(y^{n})\}$ are statistically dependent.
The term required for Janson's theorem is then given by 
\begin{equation}
\Delta\dfn\sum_{x^{n}\sim y^{n}}\P\left(f\in\mathcal{B}_{\eta}(x^{n})\cap\mathcal{B}_{\eta}(y^{n})\right).\label{eq: Janson correction term}
\end{equation}
Let us analyze the probability in (\ref{eq: Janson correction term})
under the assumption that $f(x^{n})=f(y^{n})=0$. It will be evident
that all other three cases for $(f(x^{n}),f(y^{n}))$ can be analyzed
in the same way and lead to essentially the same result. Bayes rule
implies that 
\begin{align}
 & \P\left(f\in\mathcal{B}_{\eta}(x^{n})\cap\mathcal{B}_{\eta}(y^{n})\mid f(x^{n})=f(y^{n})=0\right)\nonumber \\
 & =\P\left(f\in\mathcal{B}_{\eta}(x^{n},1)\mid f(x^{n})=f(y^{n})=0\right)\nonumber \\
 & \hphantom{=}\times\P\left(f\in\mathcal{B}_{\eta}(y^{n},1)\mid f(x^{n})=f(y^{n})=0,f\in\mathcal{B}_{\eta}(x^{n},1)\right)\nonumber \\
 & \hphantom{=}\times\P\Bigg(\bigcap_{d=2}^{\ell}\left\{ \left\{ f\in\mathcal{B}_{\eta}(x^{n},d)\right\} \cap\left\{ f\in\mathcal{B}_{\eta}(y^{n},d)\right\} \right\} \mid\nonumber \\
 & \hphantom{=}\hphantom{\hphantom{=}\hphantom{=}\hphantom{=}\hphantom{=}=}f(x^{n})=f(y^{n})=0,f\in\mathcal{B}_{\eta}(x^{n},1)\cap\mathcal{B}_{\eta}(y^{n},1)\Bigg).\label{eq: Bayes rule first derivation}
\end{align}
For the first probability on the r.h.s. of (\ref{eq: Bayes rule first derivation}),
we note that if $\Hamd(x^{n},y^{n})\geq2$ then $\mathcal{S}(x^{n},1)\cap\{x^{n},y^{n}\}=\phi$
and (\ref{eq: M1 large value asymptotic probability}) holds. Otherwise,
if $\Hamd(x^{n},y^{n})=1$ then $\mathcal{S}(x^{n},1)\cap\{x^{n},y^{n}\}=y^{n}$.
In that case, 
\begin{align}
 & \P\left(f\in\mathcal{B}_{\eta}(x^{n},1)\mid f(x^{n})=f(y^{n})=0\right)\nonumber \\
 & =\P\left(\beta_{1,f}(x^{n})\geq1-\eta\mid f(x^{n})=f(y^{n})=0\right)\nonumber \\
 & =\P\left(\frac{1}{\binom{n}{1}}\sum_{\tilde{y}^{n}\in\mathcal{S}(x^{n},1)\backslash\{y^{n}\}}f(\tilde{y}^{n})+f(y^{n})\geq1-\eta\mid f(x^{n})=f(y^{n})=0\right)\nonumber \\
 & =\P\left(\frac{1}{n-1}\sum_{\tilde{y}^{n}\in\mathcal{S}(x^{n},1)\backslash\{y^{n}\}}f(\tilde{y}^{n})\geq\frac{n}{n-1}(1-\eta)\right)\nonumber \\
 & =2^{-(n-1)\left[1-\binent\left(\eta+O(n^{-1})\right)\right]+o(n)}\nonumber \\
 & =2^{-n(1-\binent(\eta))+o(n)},\label{eq: first layer A step - derivation}
\end{align}
where the last transition is since $\binent(\eta)$ is a smooth function,
with bounded derivatives around a neighborhood of any fixed $\eta\in(0,1)$.

For the second probability on the r.h.s. of (\ref{eq: Bayes rule first derivation}),
if $\Hamd(x^{n},y^{n})\geq3$ then $\mathcal{S}(y^{n},1)\cap\{\{x^{n},y^{n}\}\cup\mathcal{S}(x^{n},1)\}=\phi$
and (\ref{eq: M1 large value asymptotic probability}) holds. Next,
if $\Hamd(x^{n},y^{n})=1$ then $\mathcal{S}(y^{n},1)\cap\{\{x^{n},y^{n}\}\cup\mathcal{S}(x^{n},1)\}=x^{n}$.
A derivation similar to (\ref{eq: first layer A step - derivation})
shows that 
\begin{equation}
\P\left(f\in\mathcal{B}_{\eta}(y^{n},1)\mid f(x^{n})=f(y^{n})=0,f\in\mathcal{B}_{\eta}(x^{n},1)\right)=2^{-n(1-\binent(\eta))+o(n)}\label{eq: first layer B step}
\end{equation}
holds. If $\Hamd(x^{n},y^{n})=2$ then $\mathcal{S}(y^{n},1)\cap\{\{x^{n},y^{n}\}\cup\mathcal{S}(x^{n},1)\}$
contains exactly two points. Again, a derivation similar to (\ref{eq: first layer A step - derivation})
(with $n-2$ replacing $n-1$) shows that (\ref{eq: first layer B step})
holds. 

The third probability in the r.h.s. of (\ref{eq: Bayes rule first derivation})
can be trivially upper bounded by $1$. Thus,

\[
\P\left(f\in\mathcal{B}_{\eta}(x^{n})\cap\mathcal{B}_{\eta}(y^{n})\mid f(x^{n})=f(y^{n})=0\right)\leq2^{-2n(1-\binent(\eta))+o(n)}.
\]
Evidently, analogous analysis holds for all other three possibilities
of the pair $(f(x^{n}),f(y^{n}))$,\footnote{Note that the value of $f(x^{n})$ (resp. $f(y^{n})$) does not change
the asymptotics of the $\P(f\in\mathcal{B}_{\eta}(y^{n}))$ (resp.
$\P(f\in\mathcal{B}_{\eta}(x^{n}))$).} and so 
\begin{equation}
\P\left(f\in\mathcal{B}_{\eta}(x^{n})\cap\mathcal{B}_{\eta}(y^{n})\right)\leq2^{-2n(1-\binent(\eta))+o(n)}.\label{eq: bound on pairwise bad events}
\end{equation}
Now, the number of dependent pairs is upper bounded by $2^{n}\cdot\binom{n}{2\ell}$
since $x^{n}\sim y^{n}$ is possible only when $\Hamd(x^{n},y^{n})\leq2\ell$.
As $\ell=\log n$ was chosen, $\binom{n}{2\ell}\leq n^{\log n}=2^{\log^{2}n}$.
Then (\ref{eq: bound on pairwise bad events}) implies that 
\begin{align*}
\Delta & \leq2^{n+o(n)}\cdot2^{-2n(1-\binent(\eta))+o(n)}\\
 & =2^{-n\left(1-2\binent(\eta)\right)+o(n)},
\end{align*}
and so $\Delta\to0$ as $n\to\infty$ exponentially fast, as long
as $\eta\in(0,\binent^{-1}(1/2))$. This concludes the proof.
\end{proof}
We move on to the proof of Theorem \ref{thm: sharp threshold large delta}.
\begin{proof}[Proof of Theorem \ref{thm: sharp threshold large delta}]
In this proof we find it more convenient to work with a $\delta$
and $\{-1,1\}$ convention. As the proof of Theorem \ref{thm: sharp threshold delta=00003D1/n},
this proof also comprises of a local condition and global analysis.
We begin by deriving a necessary condition for a function to be $\rho$-SP,
which is now based only on the value of the function at points of
Hamming distance (slightly larger than) $2\eta n$, with $\eta<1/4$.
This condition is tailored to the regime of a fixed $\delta$. We
then use a central-limit theorem to show that the probability that
this condition is satisfied is close to $1/2$. For global analysis,
we consider a subset of the hamming cube of size about $2^{n[1-\binent(2\eta)]}$
whose minimal Hamming distance is at least $\eta n$. The existence
of such a set is guaranteed by the Gilbert-Varshamov bound \cite[Th. 4.10]{roth2006introduction}.
Since the points in this subset are sufficiently far apart, the event
that the local condition holds for one of the points is independent
of the corresponding events pertaining to all other points. Thus,
the probability of a function to be $\rho$-SP is not more than about
$2^{-2^{n[1-\binent(2\eta)]}}$. 

To prove the required local condition, let $\delta<\eta<1/4$ be given,
and let $\mathcal{D}_{\eta}(y^{n})$ be a punctured Hamming ball of
relative radius $\eta$ around $y^{n}$, i.e., 
\[
\mathcal{D}_{\eta}(y^{n})\dfn\left\{ z^{n}\in\{-1,1\}^{n}:0<\frac{1}{n}\Hamd(z^{n},y^{n})\leq\eta\right\} .
\]
Then, clearly 
\[
\left|p(x^{n}|y^{n})\cdot f(x^{n})\right|\leq p(x^{n}|y^{n})=2^{-n\log\frac{1}{1-\delta}},
\]
and by the Chernoff bound (or the method of types \cite{csiszar2011information})
\begin{align*}
\left|\sum_{x^{n}\in\mathcal{D}_{\eta}^{c}(y^{n})\backslash y^{n}}p(x^{n}|y^{n})\cdot f(x^{n})\right|\leq & \sum_{x^{n}\in\mathcal{D}_{\eta}^{c}(y^{n})\backslash y^{n}}p(x^{n}|y^{n})\\
 & \leq\P\left(X^{n}\not\in\mathcal{D}_{\eta}(y^{n})\mid Y^{n}=y^{n}\right)\\
 & \leq2^{-n\bindiv(\eta||\delta)-\Theta(\log n)}.
\end{align*}
Focusing on some $y^{n}$, we may assume without loss of generality
that $f(y^{n})=-1$. Then,
\begin{align*}
 & \E\left(f(X^{n})\mid Y^{n}=y^{n}\right)\\
 & =\sum_{x^{n}}p(x^{n}|y^{n})\cdot f(x^{n})\\
 & =p(y^{n}|y^{n})\cdot f(y^{n})+\sum_{x^{n}\in\mathcal{D}_{\eta}(y^{n})}p(x^{n}|y^{n})\cdot f(x^{n})+\sum_{x^{n}\in\mathcal{D}_{\eta}^{c}(y^{n})\backslash y^{n}}p(x^{n}|y^{n})\cdot f(x^{n})\\
 & \geq-2^{-n\log\frac{1}{1-\delta}}+\sum_{x^{n}\in\mathcal{D}_{\eta}(y^{n})}p(x^{n}|y^{n})\cdot f(x^{n})-2^{-n\bindiv(\eta||\delta)-\Theta(\log n)},
\end{align*}
and thus,
\begin{multline*}
\left\{ f\text{ is }\text{\ensuremath{\rho}}\text{-SP at }y^{n}\right\} \subseteq\\
\left\{ \sum_{x^{n}\in\mathcal{D}_{\eta}(y^{n})}p(x^{n}|y^{n})\cdot f(x^{n})\leq2^{-n\log\frac{1}{1-\delta}}+2^{-n\bindiv(\eta||\delta)-\Theta(\log n)}\right\} \dfn\mathcal{A}_{\eta}(y^{n}).
\end{multline*}
We next evaluate the probability that the necessary condition is satisfied
when $f$ is chosen uniformly at random over all Boolean functions
on $\{-1,1\}^{n}$. Specifically, we use the Berry-Esseen central-limit
theorem \cite[Chapter XVI.5, Theorem 2]{Feller} to bound $\P(\mathcal{A}_{\eta}(y^{n}))$.
To that end, we note that 
\[
\E\left(\sum_{x^{n}\in D_{\eta}(y^{n})}p(x^{n}|y^{n})\cdot f(x^{n})\right)=0,
\]
and that by the method of types \cite{csiszar2011information}
\begin{align}
\E\left(\sum_{x^{n}\in\mathcal{D}_{\eta}(y^{n})}p(x^{n}|y^{n})\cdot f(x^{n})\right)^{2} & =\sum_{x^{n}\in\mathcal{D}_{\eta}(y^{n})}p^{2}(x^{n}|y^{n})\nonumber \\
 & =\sum_{\ell=1}^{\lfloor\eta n\rfloor}\sum_{x^{n}\colon\Hamd(x^{n},y^{n})=\ell}2^{-2n[\binent(\ell/n)+\bindiv(\ell/n||\delta)]}\nonumber \\
 & =2^{-2n\cdot\min_{0\leq\zeta\leq\eta}[\binent(\zeta)+2\bindiv(\zeta||\delta)]-\Theta(\log n)}\label{eq: variance for random function}\\
 & =2^{-n\cdot\log\frac{1}{\delta^{2}+(1-\delta)^{2}}-\Theta(\log n)}\nonumber 
\end{align}
where $\zeta\dfn d/n$, and the minimum in (\ref{eq: variance for random function})
is attained for $\zeta=\frac{\delta^{2}}{\delta^{2}+(1-\delta)^{2}}$
(which satisfies $\zeta\leq\delta<\eta$). Similarly, we note that
\begin{align*}
\gamma_{n} & \dfn\sum_{x^{n}\in\mathcal{D}_{\eta}(y^{n})}\E\left|p(x^{n}|y^{n})\cdot f(x^{n})\right|^{3}\\
 & =2^{-n\cdot\min_{0\leq\zeta\leq\eta}[2\binent(\zeta)+3\bindiv(\zeta||\delta)]-\Theta(\log n)},
\end{align*}
where clearly $\gamma_{n}$ decreases exponentially for any $\delta\in(0,1/2)$.
Consequently, the Berry-Esseen central-limit theorem implies that
there exists a universal constant $c$ such that 
\begin{align}
\P\left(\mathcal{A}_{\eta}(y^{n})\right) & \leq1-Q\left(\frac{2^{-n\log\frac{1}{1-\delta}}+2^{-n\bindiv(\eta||\delta)-\Theta(\log n)}}{2^{-n\cdot\frac{1}{2}\log\frac{1}{\delta^{2}+(1-\delta)^{2}}-\Theta(\log n)}}\right)+c\gamma_{n}\label{eq: sharp threshold fixed delta Q function}\\
 & \leq\frac{1}{2}+o(1).\label{eq: sharp threshold fixed delta}
\end{align}
where in (\ref{eq: sharp threshold fixed delta Q function}) $Q(\cdot)$
is the tail distribution function of the standard normal distribution,
and (\ref{eq: sharp threshold fixed delta}) is satisfied whenever
$\eta>\eta_{\delta}$. This completes the analysis of the local necessary
condition. 

We next move on to global analysis. By the Gilbert-Varshamov bound
\cite[Th. 4.10]{roth2006introduction}, there exists a set (also known
as an\emph{ error-correcting code}) $\mathcal{C}_{n}\subset\{-1,1\}^{n}$
such that 
\[
|\mathcal{C}_{n}|\geq2^{n[1-\binent(2\eta)]-o(n)}
\]
and $\mathcal{\mathcal{D}}_{\eta}(x^{n})\cap\mathcal{\mathcal{D}}_{\eta}(y^{n})=\phi$
for all $x^{n},y^{n}\in\mathcal{C}_{n}$. Consequently, 
\begin{align*}
\P\left(f\text{ is }\text{\ensuremath{\rho}}\text{-SP}\right) & \leq\P\left(\bigcap_{y^{n}\in\mathcal{C}_{n}}f\text{ is }\text{\ensuremath{\rho}}\text{-SP at }y^{n}\right)\\
 & \leq\P\left(\bigcap_{y^{n}\in\mathcal{C}_{n}}\left\{ f\in\mathcal{A}_{\eta}^{c}(y^{n})\right\} \right)\\
 & =\prod_{y^{n}\in\mathcal{C}_{n}}\P\left(f\in\mathcal{A}_{\eta}(y^{n})\right)\\
 & \leq\left(\frac{1}{2}+o(1)\right)^{|\mathcal{C}_{n}|}.
\end{align*}
The proof is completed since $|\mathcal{C}_{n}|$ increases exponentially
for $\eta<1/4$. 
\end{proof}
\begin{rem}
It is evident that the proof of Theorem \ref{thm: sharp threshold large delta}
also holds for any sequence of $\{\delta_{n}\}$ such that $\delta_{n}=\omega(\frac{1}{\sqrt{n}})$
and $\overline{\delta}\dfn\limsup_{n\to\infty}\delta_{n}<\delta_{\max}$.
Indeed, (\ref{eq: sharp threshold fixed delta}) holds in this case
too, as long as we choose $\eta$ to be $\eta_{\overline{\delta}}$. 
\end{rem}

\section{Open Problems \label{sec:Open Problems}}

We have introduced the notion of self-predictability for Boolean functions.
There are many interesting questions left open; below is far from
an exhaustive list.

We know that the characters, Majority and a few other LTFs (found
numerically) are USP, and we can create many other USP functions from
them. However, we still lack a clear understanding of what makes a
function USP.
\begin{problem}
Characterize the family of USP functions. Specifically, how many USP
functions are there?
\end{problem}
More specifically, we ask: 
\begin{problem}
Is there a finite set of USP functions and a finite set of SP-preserving
operations that span all USP functions? 
\end{problem}
Adding symmetry to the mix, we conjecture the following. 
\begin{conjecture}
The only symmetric USP functions are Majority and the largest character.
In particular, Majority is the only monotone and symmetric USP function. 
\end{conjecture}
We have seen that LCSP functions are WST, but not vice versa. 
\begin{problem}
Find a simple condition guaranteeing that a WST function is LCSP (resp.
USP).
\end{problem}
We say that a function $f$ is \emph{monotonically SP} if there exists
$\rho_{0}$ such that $f$ is $\rho$-SP for $\rho>\rho_{0}$ and
not $\rho$-SP for $\rho<\rho_{0}$. We have seen that there exist
(balanced) functions that are not monotonically SP. 
\begin{problem}
Characterize the family of monotonically SP functions. 
\end{problem}
We have bounded the ratio between the strongest and second strongest
coefficient of an LCSP LTF. This is quite weak: Let $r_{n}(\mathcal{F})$
be the minimum number such that any LTF in the family $\mathcal{F}$
admits a representation in which the ratio between the maximal coefficient
and minimal coefficient (in absolute values) is at most $r_{n}(\mathcal{F})$.
It is known in general, (see \cite[Theorem 2]{babai2010weights} and
references therein) that $2^{-n(2-o(1))}\cdot n^{n/2}\leq r_{n}(\mathcal{F})\leq2^{n-1}\cdot(n+1)^{(n+1)/2}$
if $\mathcal{F}$ is the family of all LTFs. It is interesting to
ask whether $r_{n}(\mathcal{F})$ becomes much smaller under self-predictability. 
\begin{problem}
Characterize $r_{n}(\mathcal{F})$ when $\mathcal{F}$ is the family
of LCSP LTFs (resp. USP LTFs). 
\end{problem}
Let $\mathcal{G}_{\rho,n}$ be a directed graph over the set of all
Boolean functions with $n$ variables, where we draw a directed edge
from every function $f$ to its optimal predictor $\sgn T_{\rho}f$
(unless they coincide). To avoid ambiguities, we can set $\sgn T_{\rho}f$
equal to $f$ whenever $T_{\rho}f$ is exactly zero. It is easy to
see that the number of $\rho$-SP functions is upper bounded by the
number of weakly connected components of $\mathcal{G}_{\rho,n}$,
namely the connected components of the associated undirected graph
obtained by removing the direction of the edges. In fact, we conjecture
that these quantities are exactly equal, or equivalently: 
\begin{conjecture}
$\mathcal{G}_{\rho,n}$ contains no cycles.
\end{conjecture}
Note that if the above conjecture holds, then a simple way to arrive
at a $\rho$-SP function is to start with some function $f$ and repeatedly
apply the $\sgn T_{\rho}$ operator; this procedure will terminate
at a $\rho$-SP function in finite time. In fact, simulations indicate
that this convergence happens very quickly, which may hint that the
weakly connected components of $\mathcal{G}_{\rho,n}$ have small
depth. 

\section*{Acknowledgments}

We are grateful to Or Ordentlich for coming up with the original (different)
argument that Majority is USP, based on May\textquoteright s Theorem
(Remark \ref{rem: proof of majority is USP via May's theorem}), and
to Lele Wang for coming up with Example \ref{exa:balanced function may not have balanced predictor}.
We would like to thank both, as well as Omri Weinstein, for providing
valuable insight during many stimulating discussions.

\bibliographystyle{siam}
\bibliography{Self_Predicting_Functions}

\end{document}